\documentclass[a4paper,UKenglish,cleveref,autoref, thm-restate]{lipics-v2021}

\title{Finding Diverse Strings and Longest Common Subsequences in a Graph}
\titlerunning{Finding Diverse Strings and Longest Common Subsequences} 

\author{Yuto Shida}{Hokkaido University, Japan}{}{}{}
\author{Giulia Punzi}{National Institute of Informatics, Japan \\ University of Pisa, Italy}
{giulia.punzi@unipi.it}
{https://orcid.org/0000-0001-8738-1595}{}
\author{Yasuaki Kobayashi}{Hokkaido University, Japan}{koba@ist.hokudai.ac.jp}{https://orcid.org/0000-0003-3244-6915}{}
\author{Takeaki Uno}{National Institute of Informatics, Japan}{uno@nii.ac.jp}{https://orcid.org/0000-0001-7274-279X}{}
\author{Hiroki Arimura}{Hokkaido University, Japan}{arim@ist.hokudai.ac.jp}{https://orcid.org/0000-0002-2701-0271}{}
\authorrunning{
  Y.~Shida,
  G.~Punzi, 
  Y.~Kobayashi, 
  T.~Uno, and 
  H.~Arimura 
}

\Copyright{
  Yuto Shida,
  Giulia Punzi, 
  Yasuaki Kobayashi, 
  Takeaki Uno, and 
  Hiroki Arimura, 
} 

\category{} 


\relatedversion{}

\funding{
Uno Takeaki: Work supported by JSPS KAKENHI Grant
Number JP20H05962.  
Yasuaki Kobayashi: Supported by JSPS KAKENHI Grant
Numbers JP20H00595, JP23K28034, JP24H00686, and JP24H00697.  Hiroki
Arimura: Supported by JSPS KAKENHI Grant Numbers JP20H00595,
JP20H05963, and JST CREST Grant Number JPMJCR18K2.
Giulia Punzi: Work supported by JSPS KAKENHI Grant
Number JP20H05962 and partially supported by MIUR, the Italian Ministry of Education, University and Research, under Grant 20174LF3T8 AHeAD: efficient Algorithms for HArnessing networked Data.
}

\EventEditors{Shunsuke Inenaga and Simon J. Puglisi}
\EventNoEds{2}
\EventLongTitle{35th Annual Symposium on Combinatorial Pattern Matching (CPM 2024)}
\EventShortTitle{CPM 2024}
\EventAcronym{CPM}
\EventYear{2024}
\EventDate{June 25--27, 2024}
\EventLocation{Fukuoka, Japan}
\EventLogo{}
\SeriesVolume{296}
\ArticleNo{21}

\usepackage[appendix=append]{apxproof} 

\usepackage{graphicx}
\usepackage{amsmath,amssymb}
\usepackage{mathtools}
\mathtoolsset{showonlyrefs=false} 
\usepackage{cite}
\usepackage {booktabs}
\usepackage[skip=0.5ex]{subcaption} 
\usepackage{url}
\usepackage{bm}
\usepackage{bbm}
\usepackage{textcomp}

\usepackage{tikz}
\usetikzlibrary{backgrounds, positioning}

\newcommand{\myparagraph}[1]{\textbf{#1}.\hspace{0.25em}}

\newsavebox{\cmbox}
\newenvironment{commbox}{
  \begin{lrbox}{\cmbox}
    \begin{minipage}{.9\textwidth}
}{\end{minipage}
  \end{lrbox}
}

\newenvironment{myempty}{\begin{commbox}}{\end{commbox}}

\newenvironment{myproofonly}[1]{%
\begin{trivlist}\item[]$\blacktriangleright$\;{\sffamily\bfseries #1.}\hskip3pt}{\hfill\end{trivlist}}

\newtheorem{problem}{Problem}

\crefname{section}{Sec.}{Sections}
\crefname{chapter}{Chapter}{Chapters}
\crefname{algorithm}{Algorithm}{Algorithms}
\crefname{table}{Table}{Tables}
\crefname{figure}{Fig.}{Figures}
\crefname{definition}{Definition}{Definitions}
\crefname{lemma}{Lemma}{Lemmas}
\crefname{proposition}{Proposition}{Propositions}
\crefname{theorem}{Theorem}{Theorems}
\crefname{remark}{Remark}{Remarks}
\crefname{conjecture}{Conjecture}{Conjectures}
\crefname{problem}{Problem}{Problems}
\crefname{observation}{Observation}{Observations}
\crefname{equation}{Eq.}{Equations}

\newtheoremrep{definition}{Definition}[section]
\newtheoremrep{theorem}[theorem]{Theorem}[section]
\newtheoremrep{lemma}[theorem]{Lemma}[section]
\newtheoremrep{proposition}[proposition]{Proposition}[section]
\newtheoremrep{corollary}[theorem]{Corollary}[section]
\newtheoremrep{remark}[theorem]{Remark}
\newtheoremrep{conjecture}[theorem]{Conjecture}
\renewcommand{\paragraph}[1]{\textbf{#1}\hskip 0.5em}

\newenvironment{inalign}{\begin{math}\catcode`&=9}{\catcode`&=4\end{math}}

\newcommand{\nofbox}[1]{#1}

\renewcommand{\vec}[1]{\bm{#1}}

\newcommand{\Div}[2][d]{D_{#2}^\idrm{#1}}

\newcommand{\X}{\sig{X}}
\newcommand{\Y}{\sig{Y}}
\renewcommand{\S}{\sig{S}}

\newcommand{\Wone}[1][1]{$\fn{W}\lbrack {#1}\rbrack$}

\newcommand{\E}[1][\delta]{E^{#1}}

\newcommand{\Eout}{\E[+]}

\newcommand{\WT}{\op{Weights}}
\newcommand{\size}{\fn{size}}
\newcommand{\dep}{\fn{depth}}

\newcommand{\str}{\fn{str}}
\newcommand{\lab}{\fn{lab}}

\newcommand{\withpos}{\cup\set{0}}

\newcommand{\nat}{\mathbb{N}}
\newcommand{\rat}{\mathbb{R}}
\newcommand{\zat}{\mathbb{Z}}
\newcommand{\ratpos}{\mathbb{R}_{\ge 0}}

\renewcommand{\leq}{\leqslant}
\renewcommand{\le}{\leq}
\renewcommand{\geq}{\geqslant}
\renewcommand{\ge}{\geq}

\newcommand{\sig}[1]{\mathcal{#1}}
\newcommand{\fn}[1]{\ensuremath{\mathrm{#1}}}
\newcommand{\proc}[1]{\textsc{#1}}

\newcommand{\op}[1]{\mathtt{#1}}

\newcommand{\eps}{\varepsilon}

\def\idrm#1{\ensuremath{\mathrm{#1}}}

\newcommand{\by}{\times}
\newcommand{\pair}[1]{\langle #1\rangle}

\newcommand{\set}[1]{\{#1\}}
\newcommand{\sete}[1]{\{\:#1\:\}}

\renewcommand{\bar}[1]{\overline{#1}}

\newcommand{\idx}[3]{_{{#1}={#2}}^{#3}}
\newcommand{\rk}[1]{^{\scriptsize\kern.125pt\textrm{(#1)}}}
\def\fwd|#1|{\overrightarrow{#1}}
\def\rev|#1|{\overleftarrow{#1}}
\newcommand{\subext}[1]{_{\scriptsize\kern.125pt\textrm{#1}}}

\newcommand{\Ind}[1]{\mathbbm{1}\kern-0.2em\left\{\kern0.1em{#1}\kern0.1em\right\}}

\newcommand{\Expect}[2][{}]{\mathop{\mathbb{E}}_{#1}\kern-0.2em\left\{\rule{0pt}{2.0ex} \kern0.2em{#2}\kern0.2em\right\}}

\mathcode`\.=\string"8000
\begingroup \catcode`\.=\active
  \gdef.{\hbox{\hskip0px\raise-2pt\hbox{$\cdot$\hskip0pt}}}
\endgroup

\usepackage[ruled,vlined,linesnumbered]{algorithm2e} 

\makeatletter
\renewcommand{\@algocf@capt@plain}{above}
\makeatother
\newcommand{\brcomment}[1]{\hfill$\rhd$\ \textit{#1}}
\newcommand{\blcomment}[1]{\kern0.5em$\rhd$\ \textit{#1}}
\SetKwComment{Comment}{$\rhd$\hspace{0.125em}}{} 
\SetKwInput{KwGlobal}{Global}
\SetKwInput{KwWork}{Variable}
\SetKwInput{KwInput}{Input}
\SetKwInput{KwOutput}{Output}
\SetKwInput{KwTask}{Task}
\SetKwInput{KwProc}{Procedure}
\SetArgSty{textrm}
\SetCommentSty{textit}
\newcommand{\iIf}[2]{\textbf{if} {#1} \textbf{then}\hspace{0.125em}{\relax #2}}

\newcommand{\Procedure}{\textbf{Procedure}\hspace{0.375em}}

\newcommand{\citep}[2][]{{#1}\cite{#2}}

\usepackage[appendix=append]{apxproof}

\nolinenumbers

\begin{document}
\maketitle



\newcount\pagefighardmatch 
\newcount\pagefigfptalgo 


\begin{abstract}
In this paper, we study for the first time the \textit{Diverse Longest Common Subsequences} (LCSs) problem under Hamming distance. Given a set of a constant number of input strings, the problem asks to decide if there exists some subset $\mathcal X$ of $K$ longest common subsequences whose \textit{diversity} is no less than a specified threshold $\Delta$, where we consider two types of diversities of a set $\mathcal X$ of strings of equal length: the \textit{Sum diversity} and the \textit{Min diversity} defined as the sum and the minimum of the pairwise Hamming distance between any two strings in $\mathcal X$, respectively. We analyze the computational complexity of the respective problems with Sum- and Min-diversity measures, called the \textit{Max-Sum} and \textit{Max-Min Diverse LCSs}, respectively, considering both approximation algorithms and parameterized complexity. Our results are summarized as follows. When $K$ is bounded, both problems are polynomial time solvable. In contrast, when $K$ is unbounded, both problems become NP-hard, while Max-Sum Diverse LCSs problem admits a PTAS. Furthermore, we analyze the parameterized complexity of both problems with combinations of parameters $K$ and $r$, where $r$ is the length of the candidate strings to be selected. Importantly, all \textit{positive results} above are proven in a more general setting, where an input is an edge-labeled directed acyclic graph (DAG) that succinctly represents a set of strings of the same length. \textit{Negative results} are proven in the setting where an input is explicitly given as a set of strings. The latter results are equipped with an encoding such a set as the longest common subsequences of a specific input string set.  
\end{abstract}

\setcounter{page}{0}
\newpage
\section{Introduction}
\label{sec:intro}


The problem of finding a \textit{longest common subsequence} (LCS) of a set of $m$ strings, called the LCS problem, is a fundamental problem in computer science, extensively studied in theory and applications for over fifty years~\cite{sankoff1972matching,hirschberg1983recent,maier1978complexity,bodlaender1995parameterized,irving1992mlcs:two}.
In application areas such as computational biology, pattern recognition, and data compression, longest common subsequences are used for consensus pattern discovery and multiple sequence alignment~\cite{sankoff1972matching,gusfield1993efficient}. It is also common to use the length of longest common subsequence as a similarity measure between two strings.
For example, \cref{fig:example:lcs} shows longest common subsequences (underlined) of the input strings $X_1 = ABABCDDEE$ and $Y_1 = ABCBAEEDD$.

\newcommand{\emphlcs}[1]{\underline{#1}}
\begin{table}[h]
\vspace{-0.0\baselineskip}
\abovecaptionskip=.0\baselineskip
\caption{Longest common subsequences of two input strings $X_1$ and $Y_1$ over $\Sigma = \set{A,B,C,D,E}$.}\label{fig:example:lcs}
\centering
\small
\fbox{
\parbox{.95\textwidth}{
\centering
$\epsilon,\; A,\; B,\; C,\; D,\; E$,\; 
$AA,\; AB,\; AC,\; AD,\; AE,\; BA,\; \dots,\; CD,\; CE,\; DD,\; EE$,\;\\
$ABA,\; ABB,\; ABC,\; ABD,\; \dots,\; CEE,\;$
$ABAD,\; ABAE,\; ABBD,\; \dots,\; BCEE$,\; \\
$\emphlcs{ABADD},\; \emphlcs{ABAEE},\; \emphlcs{ABBDD},\;$
    $\emphlcs{ABBEE},\; \emphlcs{ABCDD},\; \emphlcs{ABCEE}$
}
}
\end{table}

The LCS problem can be solved in polynomial time for constant $m \ge 2$ using dynamic programming by Irving and Fraser~\cite{irving1992mlcs:two} requiring $O(n^m)$ time, where $n$ is the maximum length of $m$ strings. When $m$ is unrestricted, LCS is NP-complete~\cite{maier1978complexity}.
From the view of parameterized complexity,
Bodlaender, Downey, Fellows, and Wareham~\cite{bodlaender1995parameterized} showed that 
the problem
is W[$t$]-hard parameterized with $m$ for all $t$, 
is W[2]-hard parameterized with the length $\ell$ of a longest common subsequence, and
is W[1]-complete parameterized with $\ell$ and $m$.
Bulteau, Jones, Niedermeier, and Tantau~\cite{bulteau2022fpt} presented a \textit{fixed-paraemter tractable} (FPT) algorithm with different parameterization. 

Recent years have seen increasing interest in efficient methods for finding a \textit{diverse set of solutions}~\cite{baste2022diversity,fellows:rosamond2019algorithmic,hanaka2023framework,matsuoka2024maximization}. 
Formally, let $(\sig F, d)$ be a distance space with a set $\sig F$ of feasible solutions and a distance $d: \sig F\by\sig F\to \ratpos$, where $d(X, Y)$ denotes the distance between two solutions $X, Y \in \sig F$. We consider two diversity measures for
a subset
$\X = \set{X_1, \dots, X_K}\subseteq \sig F$ of solutions: 
\begin{align}
  \Div[sum]{d}(\X) &:=\textstyle  \sum_{i<j} d(X_i, X_j),
  & (\textsc{Sum diversity}), 
    \\ \Div[min]{d}(\X) &:= \textstyle \min_{i<j} d(X_i, X_j),
    & (\textsc{Min diversity}). 
\end{align}
For $\tau \in \set{\fn{sum},\fn{min}}$, 
a subset $\X\subseteq \sig F$ of feasible solutions is said to be \textit{$\Delta$-diverse} w.r.t.~$\Div[\tau]{d}$ (or simply, \textit{diverse}) if $\Div[\tau]{d}(\X)\ge \Delta$ for a given $\Delta\ge 0$. 
Generally, the \textsc{Max-Sum} (resp.~\textsc{Max-Min}) \textsc{Diverse Solutions} problem related to a combinatorial optimization problem $\Pi$ is the problem of, given an input $I$ to $\Pi$ and a nonnegative number $\Delta \ge 0$, deciding if there exists a subset $\sig X 
\subseteq \fn{Sol}_\Pi(I)$ of $K$ solutions on $I$ such that $\Div[sum]{d}(\sig X) \ge \Delta$ (resp.~$\Div[min]{d}(\sig X) \ge \Delta$), 
where $\fn{Sol}_\Pi(I)\subseteq \sig F$ is the set of solutions on $I$. 
For many distance spaces related to combinatorial optimization problems, both problems are known to be computationally hard with unbounded~$K$~\cite{baste2022diversity,baste2019fpt,chandra2001approximation,erkut1990discrete,fellows:rosamond2019algorithmic,hanaka2021finding,hanaka2023framework,hansen1988dispersing,kuby1987programming,wang1988study}.

In this paper, we consider the problem of finding
a diverse set of solutions for 
\textit{longest common subsequences} of a set $\S$ of input strings under Hamming distance.
The task is to select $K$ longest common subsequences, maximizing the minimum pairwise Hamming distance among them. In general, a set of $m$ strings of length $n$ may have exponentially many longest common subsequences in $n$. Hence, efficiently finding such a diverse subset of solutions for longest common subsequences is challenging. 

Let $d_H(X, Y)$ denote the Hamming distance between two strings $X, Y \in \Sigma^r$ of the equal length $r\ge 0$, called \textit{$r$-strings}.
Throughout this paper, we consider two diversity measures over sets of equi-length strings, the Sum-diversity $\Div[sum]{d_H}$ and the Min-diversity $\Div[min]{d_H}$ under the Hamming distance $d_H$. 
Let $LCS(\S)$ denotes the \textit{set of all longest common subsequences} of a set $\S$ of strings. 
Now, we state our first problem.

\begin{problem}\label{def:problem:maxmin:lcs}
\rm\parindent=0mm
\textsc{Diverse LCSs with Diversity Measure $\Div[\tau]{d_H}$}
\par\textit{Input}: Integers $K, r\ge 1$, and $\Delta \ge 0$, and a set $\S = \set{S_1, \dots, S_m}$ of $m\ge 2$ strings over $\Sigma$ of length at most $r$ ; 
\par\textit{Question}: Is there some set $\X\subseteq LCS(\S)$
of longest common subsequences of $\S$
such that $|\X| = K$ and $\Div[\tau]{d_H}(\X) \ge \Delta$?
\end{problem}

Then, we analyze the computational complexity of \textsc{Diverse LCSs} from the viewpoints of approximation algorithms~\cite{vazirani:book2010approx} and parameterized complexity~\cite{cygan:book2015parameterized,flum:grohe:book2006param:theory}.
For proving positive results for the case that $K$ is bounded, actually, we work with a more general setting in which a set of strings to select is implicitly represented by the language $L(G)$ accepted by an edge-labeled DAG $G$, called a $\Sigma$-DAG. This is motivated by the fact implicit within the well-known algorithm for $K$-LCSs by Irving and Fraser~\cite{irving1992mlcs:two} that the set $LCS(\sig S)$ can be succinctly represented by such a $\Sigma$-DAG (see \cref{lem:lcsdag}). 
In contrast, \textit{negative results} will be proven in the setting where
an input is explicitly given as a set of strings.

Let $\tau \in \set{\fn{sum},\fn{min}}$ be any diversity type. Below, we state the modified version of the problem, where an input string set is an arbitrary set of equi-length strings, no longer a set of LCSs, and it is implicitly represented by 
either a $\Sigma$-DAG $G$ or the set $L$ itself.

\begin{problem}\label{def:problem:maxmin:sigdag}
\rm\parindent=0mm
\textsc{Diverse String Set with Diversity Measure $\Div[\tau]{d_H}$}\par 
\textit{Input}:
Integers $K$, $r$, and $\Delta$, and a $\Sigma$-DAG $G$ for a set $L(G) \subseteq \Sigma^r$ of $r$-strings. \par 
\textit{Question}: Decide if there exists some subset $\X\subseteq L(G)$ such that $|\X| = K$ and $\Div[\tau]{d_H}(\X)\ge \Delta$. 
\end{problem}

\textbf{Main results}.
Let $K\ge 1$, $r > 0$, and $\Delta\ge 0$ be integers and $\Sigma$ be an alphabet. The underlying distance is always Hamming distance $d_H$ over $r$-strings.
In \textsc{Diverse String Set}, we assume that an input string set $L\subseteq \Sigma^r$ of $r$-strings is represented by either a $\Sigma$-DAG $G$ or the set $L$ itself.
In \textsc{Diverse LCS}, we assume that the number $m = |\S|$ of strings in an input set $\S$ is assumed to be constant throughout. 
Then, the main results of this paper are summarized as follows.
\begin{enumerate}[1.]
\item When $K$ is bounded, both \textsc{Max-Sum} and \textsc{Max-Min} versions of \textsc{Diverse String Set} and \textsc{Diverse LCSs} can be solved in polynomial time
using dynamic programming (DP). (see \cref{lem:k:const:dp:algo:maxmin}, \cref{lem:k:const:dp:algo:maxsum})

\item When $K$ is part of the input, the \textsc{Max-Sum} version of \textsc{Diverse String Set} and \textsc{Diverse LCSs} admit a PTAS by local search showing that the Hamming distance is a \textit{metric of negative type}%
\footnote{It is a finite metric satisfying a class of inequalities of negative type~\cite{deza1997geometry:book}. For definition, see \cref{sec:approx}.}.
  (see \cref{thm:k:input:ptas})

\item Both of 
the 
\textsc{Max-Sum} and 
\textsc{Max-Min} 
versions of \textsc{Diverse String Set} and \textsc{Diverse LCSs} are fixed-parameter tractable (FPT) when parameterized by $K$ and $r$ (see \cref{thm:k:r:delta:const:fpt:maxmin}, \cref{thm:k:r:delta:const:fpt:maxsum}).
These results are shown by combining Alon, Yuster, and Zwick's \textit{color coding technique}~\cite{alon1995color} and the DP method in Result~1 above. 

\item When $K$ is part of the input, the \textsc{Max-Sum} and \textsc{Max-Min} versions of \textsc{Diverse String Set} and \textsc{Diverse LCSs} are NP-hard for any constant $r \ge 3$ (\cref{thm:nphard:div:strdag}, \cref{thm:nphard:divlcs}). 

\item When parameterized by $K$, the \textsc{Max-Sum} and \textsc{Max-Min} versions of \textsc{Diverse String Set} and \textsc{Diverse LCSs} are W[1]-hard (see \cref{thm:w1hard:div:strdag}, \cref{thm:w1hard:divlcs}). 
\end{enumerate}

A summary of these results is presented in \cref{tab:summary:results}.
We remark that the \textsc{Diverse String Set} problem coincides the original \textsc{LCS} problem when $K = 1$. It is generally believed that a W[1]-hard problem is unlikely to be $\fn{FPT}$~\cite{flum:grohe:book2006param:theory,downey:fwllows:2012parameterized}.
Future work includes the approximability of the \textit{Max-Min} version of both problems for unbounded $K$, 
and extending our results to other distances and metrics over strings, e.g., \textit{edit distance}~\cite{levenshtein1966binary,wagner1974string}
and \textit{normalized edit distance}~\cite{fisman2022normalized}. 

\begin{table}[t]
\caption{Summary of results on
\textsc{Diverse String Set}  and \textsc{Diverse LCSs} 
under Hamming distance, 
where $K$, $r$, and $\Delta$ stand for 
the \textit{number, length}, and \textit{diversity threshold} 
for a subset $\X$ of $r$-strings,
and 
$\alpha$:\;\textit{const, param}, and \textit{input} 
indicate that $\alpha$ is a constant, a parameter, and part of an input, respectively. A representation of
an input set $L$ 
is 
both of $\Sigma$-DAG and LCS otherwise stated. 
}\label{tab:summary:results}
\newcommand{\mystrut}{\rule{0pt}{1.1\baselineskip}}
\centering
\begin{tabular}{p{2.22cm}p{1.25cm}p{2.4cm}p{3.2cm}p{2.8cm}}
\toprule
Problem & Type
& $K$:\;const 
& $K$:\;param 
& $K$:\;input 
\\
\midrule
\textsc{Max-Sum \makebox{Diverse} \makebox{String Set \&\;\textsc{LCS}} }
& Exact 
& Poly-Time \makebox{(\cref{lem:k:const:dp:algo:maxsum})} 
&
\makebox{W[1]-hard on $\Sigma$-DAG} \makebox{(\cref{thm:w1hard:div:strdag})})
\makebox[2.8cm]{}
\makebox{W[1]-hard on LCS} \makebox{(\cref{thm:w1hard:divlcs})})
  &  
  \makebox{NP-hard on $\Sigma$-DAG}
  \makebox{if $r\ge 3$:const}
  \makebox{(\cref{thm:nphard:div:strdag})}  
  \makebox{NP-hard on LCS}
  \makebox{(\cref{thm:nphard:divlcs})}
\\
\cmidrule(lr){2-5}
 & Approx. or FPT
 & ---
 & \makebox{FPT if $r$:\;param} \makebox{(\cref{thm:k:r:delta:const:fpt:maxsum})}
 & PTAS \makebox{(\cref{thm:k:input:ptas})} 
 \\
\midrule
\textsc{Max-Min \makebox{Diverse} \makebox{String Set}
\&\;{LCS}}
& Exact 
& Poly-Time \makebox{(\cref{lem:k:const:dp:algo:maxmin})} 
&
\makebox{W[1]-hard on $\Sigma$-DAG}
\makebox{(\cref{thm:w1hard:div:strdag})}
\makebox[2.8cm]{}
\makebox{W[1]-hard on LCS}
\makebox{(\cref{thm:w1hard:divlcs})} 
  &
  NP-hard \makebox{on $\Sigma$-DAG} \makebox{if $r\ge 3$:const}
  \makebox{(\cref{thm:nphard:div:strdag})} 
  \makebox{NP-hard on LCS}
  \makebox{(\cref{thm:nphard:divlcs})} \\  
\cmidrule(lr){2-5}
& Approx. or FPT
& ---
& \makebox{FPT if $r$:\;param} \makebox{(\cref{thm:k:r:delta:const:fpt:maxmin})} 
& Open \rule{1.2cm}{0pt}
 \\
\bottomrule
\end{tabular}
\end{table}

\subsection{Related work}
\label{subsec:related}

\textbf{Diversity maximization for point sets} in metric space and graphs has been studied since 1970s under various names in the literature~\cite{shier1977min,kuby1987programming,erkut1990discrete,hansen1988dispersing,ravi1994heuristic,chandra2001approximation,birnbaum2009improved,cevallos2019improved} (see~Ravi, Rosenkrantz, and Tayi~\cite{ravi1994heuristic} and Chandra and Halld{\'o}rsson~\cite{chandra2001approximation} for overview).
There are two major versions: 
\textsc{Max-Min} version is known as \textit{remote-edge}, \textit{$p$-Dispersion}, and \textit{Max-Min Facility Dispersion}~\cite{shier1977min,erkut1990discrete,wang1988study}; 
\textsc{Max-Sum} version is known as \textit{remote-clique}, \textit{Maxisum Dispersion}, and \textit{Max-Average Facility Dispersion}~\cite{hansen1988dispersing,ravi1994heuristic,birnbaum2009improved,cevallos2019improved}.
Both problems are shown to be NP-hard with unbounded $K$ for general distance and metrics (with triangle inequality)~\cite{erkut1990discrete,hansen1988dispersing}, while they are polynomial time solvable for $1$- and $2$-dimensional $\ell_2$-spaces~\cite{wang1988study}. It is trivially solvable in $n^{O(k)}$ time for bounded $K$.

\textbf{Diversity maximization in combinatorial problems}. 
However, extending these results for finding diverse solutions to combinatorial problems is challenging~\cite{fellows:rosamond2019algorithmic,baste2022diversity}. While methods such as \textit{random sampling}, \textit{enumeration}, and \textit{top-$K$ optimization} are commonly used
for increasing the diversity of solution sets
in optimization, they lack theoretical guarantee of the diversity~\cite{baste2022diversity,fellows:rosamond2019algorithmic,baste2019fpt,hanaka2023framework}.
In this direction, Baste, Fellows, Jaffke, Masar{\'\i}k, de Oliveira Oliveira, Philip, and Rosamond~\cite{baste2019fpt,baste2022diversity}
pioneered the study of finding diverse solutions in combinatorial problems, investigating the parameterized complexity of well-know
graph problems such as \textit{Vertex Cover}%
~\cite{baste2019fpt}. Subsequently,
Hanaka, Kiyomi, Kobayashi, Kobayashi, Kurita, and Otachi~\cite{hanaka2021finding}
explored the fixed-parameter tractability of finding
various \textit{subgraphs}.
They further proposed a framework for \textit{approximating} diverse solutions, leading to efficient approximation algorithms
for diverse matchings, and
diverse minimum cuts~\cite{hanaka2023framework}.
While previous work has focused on diverse solutions in graphs and set families, the complexity of finding diverse solutions in \textit{string problems} remains unexplored. 
Arrighi, Fernau, de Oliveira Oliveira, and Wolf~\cite{arrighi2023synchronization} conducted one of the first studies in this direction, investigating a problem of finding a diverse set of subsequence-minimal synchronizing words. 

\textbf{DAG-based representation} for all 
LCSs 
have appeared from time to time in the literature. The LCS algorithm by Irving and Fraser~\cite{irving1992mlcs:two} for more than two strings can be seen as DP 
on a grid DAG for LCSs. 
Lu and Lin's parallel algorithm~\cite{lu1994parallel} for LCS 
used a similar grid DAG. Hakata and Imai~\cite{hakata1992longest} presented a faster algorithm based on a DAG of \textit{dominant matches}.
Conte, Grossi, Punzi, and Uno~\cite{conte:grossi:punzi:uno2023compact} and Hirota and Sakai~\cite{hirota:sakai2023enumerating:maximal} independently proposed 
DAGs 
of maximal common subsequences of two strings 
for enumeration.

\textbf{The relationship between Hamming distance and other metrics}
has been explored in string and geometric algorithms. 
Lipsky and Porat~\cite{lipsky2008l1} presented linear-time reductions 
from \textsc{String Matching} problems
under Hamming distance to equivalent problems under $\ell_1$-metric.
Gionis, Indyk, and Motwani~\cite{gionis1999similarity} used an \textit{isometry} (a distance preserving mapping) from an $\ell_1$-metric to Hamming distance over binary strings with a polynomial increase in dimension. 
Cormode and Muthukrishnan~\cite{cormode2007string} showed an efficient
$\ell_1$-embedding of edit distance allowing moves over strings into $\ell_1$-metric with small distortion. 
Despite these advancements, existing techniques haven't been successfully applied to our problems.



\section{Preliminaries}
\label{sec:prelim}

We denote by $\zat$, $\nat = \sete{ x \in \zat \mid x \ge 0 }$, $\rat$, and $\ratpos = \sete{ x \in \rat \mid x \ge 0 }$ the sets of \textit{all integers}, \textit{all non-negative integers}, \textit{all real numbers}, and \textit{all non-negative real numbers}, respectively. For any $n\in \nat$, $[n]$ denotes the set $\set{1,\dots,n}$.
Let $A$ be any set. Then, $|A|$ denotes the \textit{cardinality} of $A$. 
Throughout, our model of computation is the word RAM, where the space is measured in $\Theta(\log n)$-bit machine words. 

Let $\Sigma$ be an \textit{alphabet} of $\sigma$ symbols.
For any $n\ge 0$, $\Sigma^n$ and $\Sigma^*$ denote the sets of all strings of length $n$ and all finite strings over $\Sigma$, respectively.
Let $X = a_1\dots a_n \in \Sigma^n$ be any string. Then, the \textit{length} of $X$ is denoted by $|X| = n$. For any $1\le i,j\le n$, $X[i..j]$ denotes the substring $a_i\dots a_j$ if $i\le j$ and the \textit{empty string} $\eps$ otherwise. 
A \textit{string set} or a \textit{language} is a set $L = \set{X_1, \dots, X_n}\subseteq \Sigma^*$ of $n\ge 0$ strings  over $\Sigma$. The \textit{total length} of a string set $L$ is denoted by $||L|| = \sum_{X \in L} |X|$, while the length of the longest strings in $L$ is denoted by $\fn{maxlen}(L) := \max_{S \in L} |S|$. 
For any $r\ge 0$, we call any string $X$ an \textit{$r$-string} if its length is $r$, i.e., $X \in \Sigma^r$.
Any string set $L$ is said to be of \textit{equi-length} if $L\subseteq \Sigma^r$ for some $r\ge 0$.



\subsection{$\Sigma$-DAGs}
\label{subsec:sigma:dag}

A \textit{$\Sigma$-labeled directed acyclic graph} (\textit{$\Sigma$-DAG}, for short) 
is an edge-labeled directed acyclic graph (DAG) $G = (V, E, s, t)$ satisfying: 
 (i) $V$ is a set of vertices; 
 (ii) $E \subseteq V\by \Sigma\by V$ is a set of labeled directed edges, where 
each edge $e = (v, c, w)$ in $E$ is labeled with a symbol $c = \lab(e)$ taken from $\Sigma$; 
 (iii) $G$ has the unique \textit{source} $s$ and \textit{sink} $t$ in $V$ such that every vertex lies on a path from $s$ to $t$. 
We define the \textit{size} of $G$ as the number $\size(G)$ of its labeled edges. 
From (iii) above, $G$ contains no unreachable nodes.
For any vertex $v$ in $V$, we denote the \textit{set of its outgoing edges} by
$\Eout(v)$ $=$ $\sete{ (v, c, w) \in E \mid w \in V }$.
Any path $P = (e_1, \dots, e_n) \in E^n$ of length $n$ \textit{spells out} a string $\str(P) = \lab(e_1) \cdots \lab(e_n) \in \Sigma^n$ of length $n$, where $n\ge 0$.
A $\Sigma$-DAG $G$ \textit{represents} the string set, or language, denoted $L(G)\subseteq \Sigma^*$, as the collection of all strings spelled out by its $(s,t)$-paths. Essentially, $G$ is equivalent to a non-deterministic finite automaton (NFA)~\cite{hopcroft:ullman79book:automata} over $\Sigma$ with initial and final states $s$ and $t$, and without $\eps$-edges. 

\cref{fig:exam:main:kdivlcs:a} shows an example of $\Sigma$-DAG representing the set of six longest common subsequences of two strings in \cref{fig:example:lcs}. 
Sometimes, a $\Sigma$-DAG 
can succinctly represent a string set by its language $L(G)$. Actually, the size of $G$ can be logarithmic in $|L(G)|$ in the best case,%
\footnote{
For example, for any $r\ge 1$, the language $L = \set{a,b}^r$ over an alphabet $\Sigma=\set{a,b}$ consists of   $|L|=2^r$ strings, while it can be represented by a $\Sigma$-DAG with $2r$ edges.
}
while $\size(G)$ can be arbitrary larger than $||L(G)||$ (see \cref{lem:const:trie} in \cref{sec:fpt}). 



\begin{remarkrep}\label{rem:rep:strdag}
For any set $L$ of strings over $\Sigma$, the following properties hold: 
(1) there exists a $\Sigma$-DAG $G$ such that $L(G) = L$ and $\size(G) \le ||L||$. 
(2) Moreover, $G$ can be constructed from $L$ in $O(||L||f(\sigma))$ time, where $f(n)$ is the query time of search and insert operation on a dictionary with $n$ elements. 
(3)   Suppose that a $\Sigma$-DAG $G$ represents a set of strings $L\subseteq \Sigma^*$. 
  If $L\subseteq \Sigma^r$ for $r\ge 0$, then all paths from the source $s$ to any vertex $v$ spell out strings of the same length, say $d\le r$.
\end{remarkrep}

\begin{proof}
  (1) We can construct a \textit{trie}~$T$ for a set $L$ of strings over $\Sigma$, which is a deterministic finite automaton for recognizing $L$ in the shape of a rooted trees and has at most $O(||L||)$ vertices and edges. By identifying all leaves of $T$ to form the sink, we obtain a $\Sigma$-DAG with $||L||$ edges for $L$.
  (2) It is not hard to see that the trie $T$ can be built  in $O(||L||\log\sigma)$ time from $L$.
  (3) In what follows, we denote the string spelled out by any path $\pi$ in $G$ by $str(\pi)$. Suppose by contradiction that $G$ has some pair of paths $\pi_1$ and $\pi_2 \in E^*$ from $s$ to a vertex $v$ such that $|str(\pi_1)| - |str(\pi_2)| > 0$ (*).
  By assumption (iii) in the definition of a $\Sigma$-DAG, the vertex $v$ is contained in some $(s,t)$-path in $G$. Therefore, we have some path $\theta$ that connects $v$ to $t$. By concatenating $\pi_k$ and $\theta$, we have two $(s,t)$-paths $\tau_k = \pi_k\cdot \theta$ for all $k=1,2$.
  Then, we observe from claim (*) that
  \begin{math}
    |str(\tau_1)| - |str(\tau_2)|
    = |str(\pi_1\cdot \theta)| - |str(\pi_2\cdot \theta)|
    = (|str(\pi_1)| +  |str(\theta)|) - (|str(\pi_2)| + |str(\theta)|)
    = |str(\pi_1)| - |str(\pi_2)|
    > 0.
  \end{math} 
  On the other hand, we have $L(G)$ contains both of $str(\pi_1)$ and $str(\pi_2)$ since $\tau_1$ and $\tau_2$ are $(s,t)$-paths. This means that $L(G)$ contains two strings of distinct lengths, and this contradicts that $L(G)\subseteq \Sigma^r$ for some $r\ge 1$.
  Hence, all paths from $s$ to $v$ have the same length. Hence, (3) is proved. 
\end{proof}

By Property (3) of \cref{rem:rep:strdag}, we define the \textit{depth} of a vertex $v$ in $G$ by the length $\dep(v)$ of any path $P$ from the source $s$ to $v$, called a \textit{length-$d$ prefix (path)}. In other words, $\dep(v) = |\str(P)|$.
Then, the vertex set $V$ is partitioned into a collection of disjoint subsets $V_0 = \set{s}\cup \dots\cup V_r = \set{t}$, where $V_d$ is the subset of all vertices with depth~$d$ for all $d\in [r]\withpos$. 


\subsection{Longest common subsequences}

A string $X$ is a \textit{subsequence} of another string $Y$ if $X$ is obtained from $Y$ by removing some characters retaining the order. $X$ is a \textit{common subsequence} (CS) of
any set $\S = \set{S_1, \dots, S_m}$ of $m$ strings
if $X$ is a subsequence of any member of $\S$. A CS of $\S$ is called a \textit{longest common subsequence} (LCS) if it has the maximum length among all CSs of $\S$. 
We denote by $CS(\S)$ and $LCS(\S)$, respectively, the \textit{sets of all CSs and all LCSs} of $\S$.  
Naturally, all LCSs in $LCS(\S)$ have the same length, 
denoted by 
$0\le lcs(\S)\le \min_{S \in \S}|S|$.
While a string set $\S$ can contain exponentially many LCSs compared to the total length $||\S||$ of its strings, 
we can readily see the next lemma.

\begin{lemma}[$\Sigma$-DAG for LCSs]\label{lem:lcsdag}
  For any constant $m \ge 1$ and any set
  $\S = \set{S_1, \dots, S_m}\subseteq \Sigma^*$ of $m$ strings, 
  there exists a $\Sigma$-DAG $G$ of polynomial size in $\ell  := \fn{maxlen}(\S)$  such that $L(G) = LCS(\S)$, and $G$ can be computed in polynomial time in~$\ell$. 
\end{lemma}

\begin{proof}
  By Irving and Fraser's algorithm~\cite{irving1992mlcs:two}, we can construct  a $m$-dimensional grid graph $N$ in $O(\ell^m)$ time and space, where  
    (i) the source and sink are $s = (0,\dots,0)$ and $t = (|S_1|, \dots, |S_m|)$, respectively;   
    (ii) edge labels are symbols from $\Sigma\cup\set{\eps}$;  
    (iii) the number of edges is $\size(N) = \prod\idx i1m |S_i| \le$ $O(\ell^m)$;  
  and 
    (iv) all of $(s,t)$-paths spell out $LCS(\S)$. 
  Then, application of the $\eps$-removal algorithm~\cite{hopcroft:ullman79book:automata} yields a $\Sigma$-DAG $G$ in $O(|\Sigma|\cdot \size(N))$ time and space,
  where $G$ has $O(|\Sigma|\cdot \size(N)) = O(|\Sigma| \ell^m)$ edges.
 \end{proof}

\begin{remark}\label{lem:lcsdag:reduce}
  As a direct consequence of \cref{lem:lcsdag}, we observe that if \textsc{Max-Min} (resp.~\textsc{Max-Sum}) \textsc{Diverse String Set} can be solved in $f(M, K, r, \Delta)$ time and $g(M, K, r, \Delta)$ space on a given input DAG $G$ of size $M = \size(G)$, 
  then \textsc{Max-Min} (resp.~\textsc{Max-Sum}) \textsc{Diverse LCSs} on $\S\subseteq \Sigma^r$ can be solvable in $t = O(|\Sigma|\cdot \ell^m + f(\ell^m, K, r, \Delta))$ time and $s = O(\ell^m + g(\ell^m, K, r, \delta))$ space,
  where $\ell = \fn{maxlen}(\S)$,
  since $\size(G) = O(\ell^m)$. 
\end{remark}

From \cref{lem:lcsdag:reduce}, for any constant $m\ge 2$, there exist a polynomial time reduction from \textsc{Max-Min} (resp.~\textsc{Max-Sum}) \textsc{Diverse LCSs} for $m$ strings to \textsc{Max-Min} (resp.~\textsc{Max-Sum}) \textsc{Diverse String Set} on $\Sigma$-DAGs, where the distance measure is Hamming distance.

\subsection{Computational complexity}

A problem with parameter $\kappa$ is said to be \textit{fixed-parameter tractable} (FPT) if there is an algorithm that solves it, whose running time on an instance $x$ is upperbounded by $f(\kappa(x))\cdot |x|^c$ for a  computable function $f(\kappa)$ and constant $c > 0$.
A many-one reduction $\phi$ is called an \textit{FPT-reduction} if it can be computed in FPT and
the 
parameter $\kappa(\phi(x))$
is upper-bounded by a computable function of $\kappa(x)$. 
For notions not defined here, we refer 
to Ausiello \textit{et al.}~\cite{ausiello2012complexity} for \textit{approximability} and 
to~Flum and Grohe~\cite{flum:grohe:book2006param:theory} for \textit{parameterized complexity}. 


\section{Exact Algorithms for Bounded Number of Diverse Strings}
\label{sec:algo}

In this section, we show that both of \textsc{Max-Min} and \textsc{Max-Sum} versions of \textsc{Diverse String Set} problems can be solved by dynamic programming in polynomial time and space in the size an input $\Sigma$-DAG and integers $r$ and $\Delta$ for any constant $K$.
The corresponding results for \textsc{Diverse LCSs} problems will immediately follow from \cref{lem:lcsdag:reduce}.

\begin{figure}[t]
  \centering
    \subcaptionbox{An input $\Sigma$-DAG $G_1$\label{fig:exam:main:kdivlcs:a}}{\includegraphics[height=3.5cm]{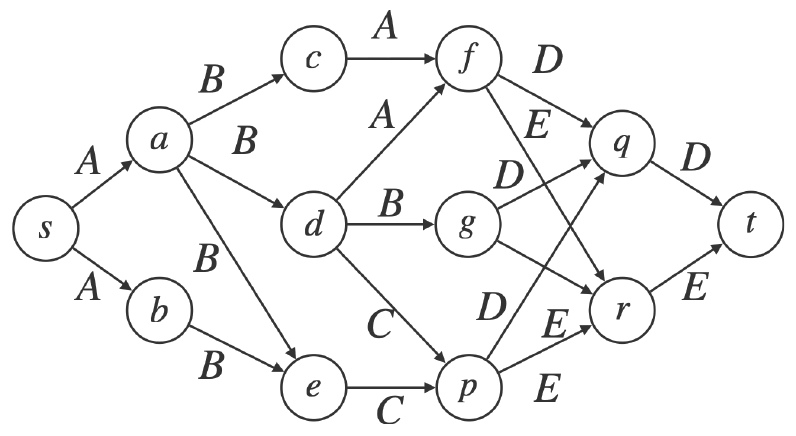}}
    \subcaptionbox{An example run of \cref{algo:k:const:dp} for $K = 3$\label{fig:exam:main:kdivlcs:b}}{\includegraphics[height=3.6cm]{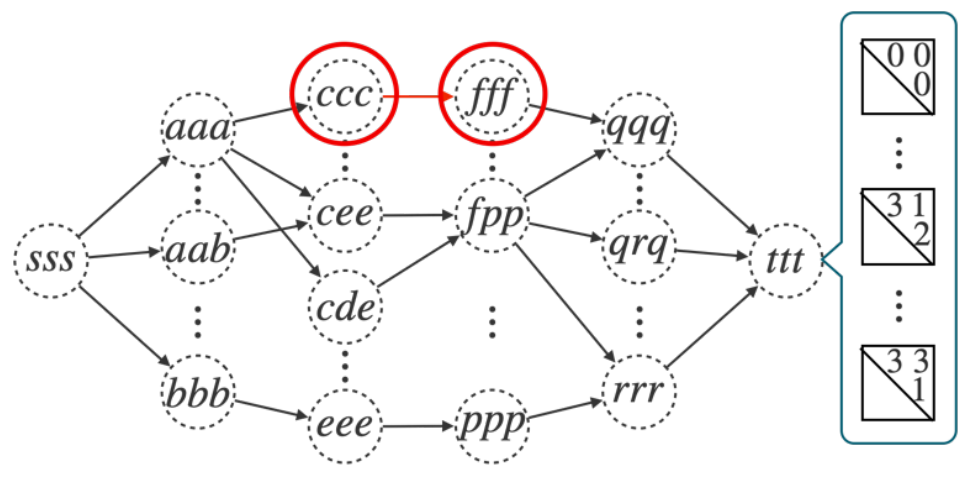}}
    \caption{Illustration of \cref{algo:k:const:dp} based on dynamic programming. 
    In (a) a $\Sigma$-DAG $G_1$ 
    represents six LCSs in \cref{fig:example:lcs}. 
    In (b), circles and arrows indicate the states of the algorithm, which are $K$-tuples of vertices of $G_1$, and transition between them, respectively. 
    All states are associated with a set of $K\by K$-weight matrices, which are shown only for the sink $ttt$ in the figure. 
    }\label{fig:exam:main:kdivlcs:all}
\end{figure}

\subsection{Computing Max-Min Diverse Solutions}
We describe our dynamic programming algorithm for the \textsc{Max-Min Diverse String Set} problem. 
Given an $\Sigma$-DAG $G = (V, E, s, t)$ with $n$ vertices, 
representing a set $L(G)\subseteq \Sigma^r$ of $r$-strings, we consider integers $\Delta\ge 0, r\ge 0$, and constant $K\ge 1$. 
A brute-force approach could solve the problem in $O(|L(G)|^K)$ time by enumerating all combinations of $K$ $(s,t)$-paths in $G$ and selecting a $\Delta$-diverse solution $\X\subseteq L(G)$. 
However, this is impractical even for constant $K$ because $|L(G)|$ can be exponential in the number of edges.

\paragraph{The DP-table.}
A straightforward method to solve the problem is enumerating all combinations of $K$-tuples of paths from $s$ to $t$ to find the best $K$-tuple. However, the number of such $K$-tuples can be exponential in $r$.
Instead, our DP-algorithm keeps track of only \textit{all possible patterns} of their pairwise Hamming distances.
Furthermore, it is sufficient to record only Hamming distance up to a given threshold $\Delta$. 
In this way, we can efficiently computes the best combination of $K$ paths provided that the number of patterns is manageable.

Formally, we let $d\;(0\le d\le r)$ be any integer and $\vec P = (P_1, \dots, P_K) \in (E^d)^K$ be any $K$-tuple of length-$d$ paths starting from the sink $s$ and ending at some nodes. Then, we define the \textit{pattern} of $K$-tuple $\vec P$ by the pair $\op{Pattern}(\vec P) = (\vec w, Z)$, where 
\begin{itemize}
\item $\vec w = (w_1, \dots, w_K) \in V_d^K$ is the $K$-tuple of
  vertices in $G$, 
called a \textit{state}, 
  such that for all $i \in [K]$, the $i$-th path $P_i$ starts from the source $s$ and ends at the $i$-th vertex $w_i$ of $\vec w$.
  
\item $Z = (Z_{i,j})_{i<j} \in [\Delta\withpos]^{K\by K}$ is an upper triangular matrix, called the \textit{weight matrix} for $\vec P$.
  For all $1\le i<j\le K$,
  $Z_{i,j} = \min(\Delta, d_{H}(\str(P_i), \str(P_j))) \in [\Delta\withpos]$
  is the Hamming distance between the string labels of $P_i$ and $P_j$ truncated by the threshold $\Delta$.
\end{itemize}

Our algorithm constructs as the DP-table 
$\WT = (\WT(\vec w, Z))_{\vec w, Z}$, which 
is a Boolean-valued table 
that associates a collection of weight matrices $Z$ to each state $\vec w$ such that $Z$ belongs to the collection if and only if $\WT(\vec w, Z) = 1$. See \cref{fig:exam:main:kdivlcs:all} for example.
Formally, we define 
$\WT$ as follows.

\begin{definition}
  $\WT: V^K \by [\Delta\withpos]^{K\by K} \to \set{0,1}$ is a Boolean table such that for every
  $K$-tuple of vertices $\vec w \in V^K $ and
  weight matrix $Z \in [\Delta\withpos]^{K\by K}$,
  $\WT(\vec w, Z) = 1$ holds if and only if $(\vec w, Z) = \op{Pattern}(\vec P)$ holds
  for some $0\le d\le r$ and some $K$-tuple
  $\vec P \in (E^d)^K$
  of length-$d$ paths from the source $s$ to $\vec w$ in $G$.
\end{definition}

We estimate the size of the table $\WT$. 
Since $Z$ takes at most 
\begin{math}
  \Gamma
  = O(\Delta^{K^2}K^2)
\end{math}
distinct values, it can be encoded in 
\begin{math}
  \log \Gamma
  = O(K^2\log\Delta)
\end{math}
bits. Therefore, $\WT$ has at most
$|V|^K\by\Gamma
= O(\Delta^{K^2}K^2M^K)
$
entries, where $M = \size(G)$. 
Consequently, for constant~$K$, $\WT$ can be stored in a multi-dimensional table of polynomial size in $M$ and $\Delta$ supporting random access in 
constant expected time or 
$O(\log\log(|V|\cdot\Delta))$ worst-case time~\cite{willard:1983yfast,cormen2022introduction}. 
%

\begin{algorithm}[t]
  \caption{An exact algorithm for solving \textsc{Max-Min Diverse $r$-String} problem. Given a $\Sigma$-DAG $G = (V, E, s, t)$ representing a set $L(G)$ of $r$-strings and integers $K\ge 1, \Delta \ge 0$, decide if there exists some $\Delta$-diverse set of $K$ $r$-strings in $L(G)$. 
  }\label{algo:k:const:dp}
Set $\WT(\vec s, Z) := 0$ for all $Z \in [\Delta\withpos]^{K\by K}$, and 
set $\WT(\vec s, \op{Zero}) \gets  1$\; 
\For{$d \gets  1, \dots, r$}{
  \For{$\vec v \gets  (v_1,\dots, v_K) \in (V_d)^K$}{
    \For{$(v_1, c_1, w_1) \in \Eout(v_1), \dots, (v_K, c_K, w_K) \in \Eout(v_K)$}{
      Set $\vec w \gets  (w_1,\dots, w_K)$\; 
      \For{$U \in [\Delta\withpos]^{K\by K}$ such that $\WT(\vec v, U) = 1$}{
        Set $Z = (Z_{i,j})_{i<j}$ with 
        $Z_{i,j} \gets  \min(\Delta, U_{i,j} + \Ind{ c_i \not= c_j }), \forall i,j \in [K]$\; 
        Set $\WT(\vec w, Z) \gets  1$ \Comment*{Update}
      }
    }
  }
}
Answer YES if $\WT(\vec t, Z) := 1$ and $\Div[min]{d_H}(Z)\ge \Delta$ for some $Z$, and NO otherwise\;
\end{algorithm}

\paragraph{Computation of the DP-table.}
We denote the $K$-tuples of copies of the source $s$ and sink $t$ by $\vec s := (s,\dots, s)$ and $\vec t := (t,\dots, t) \in V^K$, respectively, as the initial and final states. 
The \textit{zero matrix} $\op{Zero} = (\op{Zero}_{i,j})_{i<j}$ is a special matrix where $\op{Zero}_{i,j} = 0$ for all $i<j$. 
Now, we present the recurrence for the DP-table $\WT$. 

\begin{lemma}[recurrence for $\WT$]\label{lem:k:const:dp:recurrence}
  For any $0\le d\le r$,
  any $\vec w\in V^K$ and
  any $Z = (Z_{i,j})_{i<j} \in [\Delta\withpos]^{K\by K}$,
  the entry\/ $\WT(\vec w, Z) \in \set{0,1}$ satisfies the following: 
\begin{enumerate}[(1)]
\item \textsf{Base case}: If $\vec w = \vec s$ and $Z = \op{Zero}$, then $\WT(\vec w, Z) = 1$.

\item \textsf{Induction case}: If $\vec w \not= \vec s$ and all vertices in $\vec w$ have the same depth $d\;(1 \le d\le r)$, namely, $\vec w \in V_{d}^K$, then $\WT(\vec w, z) = 1$ if and only if
  there exist
  \begin{itemize}
  \item $\vec v = (v_i)\idx i1K \in V_{d-1}^K$ such that each $v_i$ is a parent of $w_i$, i.e.,
    $(v_i, c_i, w_i) \in E$, and 
  \item $U = (U_{i,j})_{i<j} \in [\Delta\withpos]^{K\by K}$ such that
      (i) $\WT(\vec v, U) = 1$, and 
      (ii) $Z_{i,j} =$ $\min(\Delta, U_{i,j} + \Ind{ c_i\not= c_j })$ for all $1\le i<j\le K$. 
\end{itemize}
  
  \item \textsf{Otherwise}: $\WT(\vec w, Z) = 0$. 
\end{enumerate}
\end{lemma}


\begin{proof}
  The proof is straightforward by induction on $0\le d\le r$. Thus, we omit the  detail.  
\end{proof}





\cref{fig:exam:main:kdivlcs:b} shows an example run of \cref{algo:k:const:dp} on a $\Sigma$-DAG $G_1$ in \cref{fig:exam:main:kdivlcs:a} representing the string set $L(G_1) = LCS(X_1, Y_1)$,
where squares indicate weight matrices. 
From \cref{lem:k:const:dp:recurrence}, we show \cref{lem:k:const:dp:algo:maxmin} on the correctness and time complexity of \cref{algo:k:const:dp}. 

\begin{theorem}[Polynomial time complexity of Max-Min Diverse String Set]\label{lem:k:const:dp:algo:maxmin}
  For any $K\ge 1$ and $\Delta \ge 0$, 
  \cref{algo:k:const:dp} solves \textsc{Max-Min Diverse String Set} in
  $O(\Delta^{K^2} K^2 M^K (\log|V| + \log\Delta))$
  time and space
  when an input string set $L$ is represented by a $\Sigma$-DAG, 
  where $M = \size(G)$ is the number of edges in $G$. 
\end{theorem}

\subsection{Computing Max-Sum Diverse Solutions}

We can solve \textsc{Max-Sum Diverse String Set} by modifying  \cref{algo:k:const:dp} as follows.
For computing the \textsc{Max-Sum} diversity, we only need to maintain the sum $z = \sum_{i<j} d_{H}(\str(P_i), \str(P_j))$ of all pairwise Hamming distances 
instead of the entire $(K\by K)$-weight matrix $Z$. 

\myparagraph{The new table $\WT'$:} 
For any $\vec w = (w_1, \dots, w_K)$ of the same depth $0\le d\le r$ and any integer
$0\le z \le rK$,
we define: 
$\WT'(\vec w, z) = 1$ if and only if there exists a $K$-tuple of length-$d$ prefix paths $(P_1, \dots, P_K) \in (E^d)^K$ from $s$ to $w_1, \dots, w_K$, respectively, such that the sum of their pairwise Hamming distances is $z$, namely,
$z = \sum_{i<j} d_{H}(\str(P_i), \str(P_j))$.

\begin{lemma}[recurrence for $\WT'$]\label{lem:k:const:dp:recurrence:maxsum}
  For any $\vec w = (w_1, \dots, w_K) \in V^K$ and any integer
  $0\le z \le rK$, 
  the entry $\WT'(\vec w, z) \in \set{0,1}$ satisfies the following: 
\begin{enumerate}[(1)]
\item \textsf{Base case}: If $\vec w = \vec s$ and $z = 0$, then $\WT(\vec w, z) = 1$.

\item \textsf{Induction case}: If $\vec w \not= \vec s$ and all vertices in $\vec w$ have the same depth $d\;(1 \le d\le r)$, namely, $\vec w \in V_{d}^K$, then $\WT(\vec w, z) = 1$ if and only if there exist
  \begin{itemize}
  \item $\vec v = (v_i)\idx i1K \in V_{d-1}^K$ such that each $v_i$ is a parent of $w_i$, i.e.,
    $(v_i, c_i, w_i) \in E$, and 
  \item $0\le u\le rK$ such that
      (i) $\WT(\vec v, u) = 1$, and 
      (ii) $z = \min(\Delta, u + \sum_{i<j}\Ind{ c_i \not= c_j })$. 
  \end{itemize}
  
  \item \textsf{Otherwise}: $\WT(\vec w, z) = 0$. 
\end{enumerate}

  
\end{lemma}

From the above modification of \cref{algo:k:const:dp} and \cref{lem:k:const:dp:recurrence:maxsum}, we have \cref{lem:k:const:dp:algo:maxsum}. From this theorem, we see that the \textsc{Max-Sum} version of \textsc{Diverse String Set} can be solved faster than the \textsc{Max-Min} version by factor of $O(\Delta^{K-1})$. 

\begin{theorem}[Polynomial time complexity of Max-Sum Diverse String Set]\label{lem:k:const:dp:algo:maxsum}
  For any constant $K\ge 1$, the modified version of \cref{algo:k:const:dp} solves \textsc{Max-Sum Diverse String Set} under Hamming Distance in
  $O(\Delta K^2 M^K(\log |V| + \log \Delta))$
  time and space, where $M = \size(G)$ is the number of edges in $G$ and the input set $L$ is represented by a $\Sigma$-DAG. 
\end{theorem}

\section{Approximation Algorithm for Unbounded Number of Diverse Strings}
\label{sec:approx}

\begin{algorithm}[t]
  \caption{A $(1 - 2/K)$-approximation algorithm for Max-Sum Diversification  for a metric $d$ of negative type on $\X$. 
}\label{algo:localsearch}
\textbf{procedure} \proc{LocalSearch}$(\sig D, K, d)$\;
$\X \gets $ arbitrary $K$ solutions in $\sig D$\;
\For{ $i \gets 1, \dots, \lceil\frac{K(K-1)}{(K+1)} \ln \frac{(K+2)(K-1)^2}{4} \rceil $ }{
  \For{$X \in \X$ such that $\sig D\setminus \X \not= \emptyset$}{
    $Y\gets
    \displaystyle
    \mathop{\fn{argmax}}_{Y \in \sig D\setminus \X}
    \hspace{0.25em}
    \mathop{\textstyle\sum}_{X' \in \X\setminus\set{X}} d(X', Y)$\;  
    $\X \gets (\X\setminus\set{X}) \cup \set{ Y }$\;    
  }
}
\fn{Output} $\X$; 
\end{algorithm}

To solve \textsc{Max-Sum Diverse String Set} for unbounded $K$,
we first introduce a local search procedure, proposed by Cevallos, Eisenbrand, and Zenklusen~\cite{cevallos2019improved}, for computing approximate diverse solutions in general finite metric spaces~(see \cite{deza1997geometry:book}). Then, we explain how to apply this algorithm to our problem in the space of equi-length strings equipped with Hamming distance. 

Let $\sig D = \set{x_1, \dots, x_n}$ be a finite set of $n\ge 1$ elements.
A semi-metric is a function $d: \sig D\by\sig D\to \ratpos$ satisfying the following conditions (i)--(iii): 
  (i) $d(x,x)=0, \forall x \in \sig D$; 
  (ii) $d(x,y)=d(y,x), \forall x,y\in \sig D$; 
  (iii) $d(x,z)\le d(x,y) + d(y,z), \forall x,y,z\in \sig D$ (triangle inequalities). 
Consider an inequality condition, called a \textit{negative inequality}: 
\begin{align}\textstyle
\vec b^{\top} D\; \vec b 
& \textstyle := \sum_{i<j} b_i b_j d(x_i, x_j) \le 0, 
& \forall \vec b = (b_1, \dots, b_n) \in \zat^n,  
\label{eq:negative:ineq}
\end{align}
where $\vec b$ is a column vector and $D = (d_{ij})$ with $d_{ij} = d(x_i, x_j)$. 
For the vector $\vec b$ above, we define $\sum \vec b := \sum_{i=1}^n b_i$.  A semi-metric $d$ is said to be \textit{of negative type} if it satisfies the inequalities \cref{eq:negative:ineq} for all $\vec b\in \zat^n$ such that $\sum \vec b = 0$.
The class $\mathbbm{NEG}$ of semi-metrics of negative type satisfies the following properties.  

\begin{lemma}[Deza and Laurent~\cite{deza1997geometry:book}] \label{lem:prop:semimetric}
For the class $\mathbbm{NEG}$, the following properties hold: 
  (1) All $\ell_1$-metrics over $\rat^r$ are semi-metrics of negative type in $\mathbbm{NEG}$ for any $r\ge 1$. 
  (2) The class $\mathbbm{NEG}$ is closed under linear combination with nonnegative coefficients in $\ratpos$. 
\end{lemma}

In \cref{algo:localsearch}, we show 
a local search procedure \textsc{LocalSearch}, proposed by Cevallos~\textit{et al.}~\cite{cevallos2019improved}, for computing a diverse solution $\X\subseteq \sig D$ with $|\X|=K$ approximately maximizing its Max-Sum diversity under a given semi-metric $d: \sig D\by\sig D\to \ratpos$ on a finite metric space $\sig D$ of $n$ points. 
The \textsc{Farthest Point} problem refer to the subproblem for computing $Y$ at Line~5. When the distance $d$ is a semi-metric of \textit{negative type}, they showed the following theorem.

\begin{theorem}[Cevallos~\textit{et al.}~\cite{cevallos2019improved}]\label{thm:cevallos2019improved}
  Suppose that $d$ is a metric of negative type over $\X$ in which the \textsc{Farthest Point} problem can be solved in polynomial time.
  For any $K\ge 1$, the procedure~\proc{LocalSearch}  in \cref{algo:localsearch} has approximation ratio $(1 - \frac{2}{K})$. 
\end{theorem}


We show that the Hamming distance actually has the desired property. 

\begin{lemma}\label{lem:hamming:negativetype}
For any integer $r\ge 1$, the Hamming distance $d_H$ over the set $\Sigma^r$ of $r$-strings is a semi-metric of negative type over $\Sigma^r$. 
\end{lemma}

\begin{proof}
  Let $\Sigma = [\sigma]$ be any alphabet. 
  We give an isometry $\phi$ (see \cref{subsec:related})
  from the Hamming distance $(\Sigma^r, d_H)$ to the $\ell_1$-metric $(W, d_{\ell_1})$ over a subset $W$ of $\rat^m$ for $m = r\sigma$. For any symbol $i \in \Sigma$, we extend $\phi$ by $\phi_\Sigma(i) := 0^{i-1}(0.5)0^{\sigma-i} \in \set{0,0.5}^\sigma$ be a bitvector with $0.5$ at $i$-th position and $0$ at other bit positions
  such that for each $c,c' \in \Sigma$, 
  \begin{math}
    ||\phi_\Sigma(c) - \phi_\Sigma(c')||_1
    = \Ind{ c \not= c'}. 
  \end{math}
  For any $r$-string $S = S[1]\dots S[r] \in \Sigma^r$, we let
  $\phi(S) := \phi_\Sigma(S[1])\cdots \phi_\Sigma(S[r])
  \in W 
  $,
  where $W := \set{0,0.5}^{m}$ and $m := r\sigma$. 
For any $S, S' \in \Sigma^r$, we can show 
\begin{inalign}
  d_{\ell_1}(\phi(S), \phi(S'))
  &= ||\phi(S)_j - \phi(S')_j||_1
  = \sum_{i\in[r]} ||\phi_\Sigma(S[i]) - \phi_\Sigma(S'[i])||_1
  &= \sum_{i\in[r]} \Ind{ S[i] \not= S'[i] }
  = d_H(S, S').
\end{inalign}
Thus, $\phi: \Sigma^r\to W$ is an \textit{isometry}~\cite{deza1997geometry:book} from $(\Sigma^r, d_H)$ to $(\set{0,0.5}^m, d_{\ell_1})$. 
By \cref{lem:prop:semimetric}, $\ell_1$-metric is a metric of negative type~\cite{deza1997geometry:book,cevallos2019improved}, and thus, the lemma is proved. 
\end{proof}

The remaining thing is efficiently solving the string version of the subproblem, called the \textsc{Farthest String} problem, that given a set $\X'\subseteq \Sigma^r$, asks to find the \textit{farthest} $Y$ from all elements in $\X'$ by maximizing the sum 
\begin{inalign}
\Div[sum]{d_H}(\X', Y) := 
\sum_{X' \in \X'} d_H(X', Y)    
\end{inalign}
over all elements $Y \in L(G)\setminus \X'$. 
For the class of $r$-strings, we show the next lemma. 

\begin{lemma}[\textsc{Farthest $r$-String}]\label{lem:fathest:sum}
  For any $K\ge 1$ and $\Delta \ge 0$, given $G$ and $\X'\subseteq L(G)$, \cref{algo:fathest:sum} computes the farthest $r$-string $Y \in L(G)$ that maximizes
  $\Div[sum]{d_H}(\X', Y)$
  over all $r$-strings in $L(G)$ in $O(K \Delta M)$ time and space, where $M$ is the number of edges in $G$. 
\end{lemma}

\begin{proof}
  The procedure in \cref{algo:fathest:sum} solves the decision version of \textsc{Max-Sum Farthest $r$-String}. Since it is obtained from \cref{algo:k:const:dp} by fixing $K-1$ paths and searching only a remaining path in $G$, its correctness and time complexity immediately follows from that of \cref{lem:k:const:dp:algo:maxmin}. It is easy to modify \cref{algo:fathest:sum} to compute such $Y$ that maximizes $\Div[sum]{d_H}(\X', Y)$ by recording the parent pair $(v, y)$ of each $(w, z)$
  and then tracing back. 
\end{proof}

\begin{algorithm}[t]
  \caption{An exact algorithm for solving the \textsc{Max-Sum Farthest $r$-String} problem. 
    Given a $\Sigma$-DAG $G$ for a set $L(G)\subseteq \Sigma^r$, a set $\X = \set{X_1, \dots, X_K}\subseteq \Sigma^r$, and an integer $\Delta \ge 0$, it decides if there exists a $Y \in L(G)$ such that $\Div[sum]{d_H}(\X, Y) \ge \Delta$. 
}\label{algo:fathest:sum}
Set $\WT(s, z) := 0$ for all $z \in [\Delta]_+$, and 
$\WT(s, 0) := 1$\; 
\For{$d := 1, \dots, r$}{
  \For{$v \in V_d$ and $(v, c, w) \in \Eout(v)$}{
      \For{$0\le u\le \Delta$ such that $\WT(v, u) := 1$}{
        Set $\WT(w, z) := 1$ for $z := u + \sum_{i \in [K]} \Ind{ c \not= X_i[d] }$ \Comment*{Update}
      }
  }
}
Answer YES if $\WT(t, \Delta) = 1$, and NO otherwise \Comment*{Decide}
\end{algorithm}

Combining \cref{thm:cevallos2019improved},  \cref{lem:fathest:sum}, and \cref{lem:hamming:negativetype}, we obtain the following theorem on the existence of a \textit{polynomial time approximation scheme} (PTAS)~\cite{ausiello2012complexity} for \textsc{Max-Sum Diverse String Set} on $\Sigma$-DAGs.
From \cref{thm:k:input:ptas} and \cref{lem:lcsdag:reduce}, the corresponding result for \textsc{Max-Sum Diverse LCSs} will immediately follow. 

\begin{theorem}[PTAS for unbounded $K$]\label{thm:k:input:ptas}
  When $K$ is part of an input, \textsc{Max-Sum Diverse String Set} problem on a $\Sigma$-DAG $G$ admits PTAS.
\end{theorem}

\begin{proof}
We show the theorem following the discussion of \cite{cevallos2019improved,hanaka2023framework}. Let $\eps > 0$ be any constant. Suppose that $\eps < 2/K$ holds. Then, $K < 2/\eps$, and thus, $K$ is a constant. In this case, by~\cref{lem:k:const:dp:algo:maxmin}, we can exactly solve the problem in polynomial time using \cref{algo:k:const:dp}. Otherwise, $2/K \le \eps$. Then, the $(1-2/K)$ approximation algorithm in \cref{algo:localsearch} equipped with \cref{algo:fathest:sum} achieves factor $1-\eps$ since $d_H$ is a negative type metric by \cref{lem:hamming:negativetype}. Hence, \textsc{Max-Sum Diverse String Set} admits a PTAS. This completes the proof. 
\end{proof}

\section{FPT Algorithms for Bounded Number and Length of Diverse Strings}
\label{sec:fpt}

In this section, we present
fixed-parameter tractable (FPT) 
algorithms for the \textsc{Max-Min} and \textsc{Max-Sum Diverse String Set}
parameterized with combinations of 
$K$ and $r$. 
Recall that a problem parameterized with $\kappa$ is said to be \textit{fixed-parameter tractable} if there exists an algorithm for the problem running on an input $x$ in time $f(\kappa(x))\cdot |x|^c$ for some computable function $f(\kappa)$ and constant $c > 0$~\cite{flum:grohe:book2006param:theory}.

For our purpose, we combine the \textit{color-coding technique} by Alon, Yuster, and Zwick~\cite{alon1995color} and
the algorithms in \cref{sec:algo}.
Consider a random $C$-coloring $c:\Sigma \to C$ from a set $C$ of $k\ge 1$ colors, which assigns a color $c(a)$ chosen from $C$ randomly and independently  to each $a\in \Sigma$. By applying this $C$-coloring to all each edges of an input $\Sigma$-DAG $G$, we obtain the $C$-colored DAG, called \textit{a $C$-DAG}, and denote it by $c(G)$. We show a lemma on reduction of $c(G)$.

\begin{lemma}[computing a reduced $C$-DAG in FPT]\label{lem:const:trie}
  For any set $C$ of $k$ colors, there exists some $C$-DAG $H$ obtained by reducing $c(G)$ such that $L(H) = L(c(G))$ and $||H|| \le k^r$. Furthermore, such a $C$-DAG $H$ can be computed from $G$ and $C$ in $t_\fn{pre} = O(k^r\cdot \size(G))$ time and space.  
\end{lemma}

\begin{proof}
  We show a proof sketch. 
  Since $L(G)\subseteq \Sigma^r$, we see that the $C$-DAG $c(G)$ represents $L(c(G)) \subseteq C^r$ of size at most $||L(c(G))|| \le k^r$. By \cref{rem:rep:strdag}, there exists a $C$-DAG $H$ for $L(H) = L(c(G))$ with at most $k^r$ edges. 
  However, it is not straightforward how to compute such a succinct $H$ directly from $G$ and $c$ in $O(k^r\cdot\size(G))$ time and space 
since $||L(G)||$ can be much larger than $k^r + \size(G)$.
We build a trie $T$ for $L(H)$ top-down using breadth-first search of $G$ from the source $s$ by maintaining a correspondence
$\varphi \subseteq V\by U$
between vertices $V$ in $G$ and vertices $U$ in $T$ (\cref{fig:fptalgo}). Then, we identify all leaves of $T$ to make the sink $t$. This runs in $O(k^r\cdot\size(G))$ time and $O(k^r + \size(G))$ space.
\end{proof}

\begin{toappendix}
  In \cref{algo:reduced:dag}, we present the procedure \proc{BuildColoredTrie} in the proof for \cref{lem:const:trie} for computation of a reduced $C$-DAG from an input $\Sigma$-DAG and a random coloring $c$.
The procedure uses a \textit{trie} for a set $L$ of strings, which is a (deterministic) finite automaton in the form of a rooted tree that represents strings in $L$ as the string labels spelled out by all paths from the root to its leaves.
  An example of the execution of the procedure will be found in \cref{fig:fptalgo} of \cref{sec:fpt}.

\begin{algorithm}[t]
  \caption{The procedure \proc{BuildColoredTrie} in \cref{lem:const:trie}.
    It computes a $C$-DAG $H$ such that $L(H) = L(c(G))$ and $||H|| \le k^r$ from a $\Gamma$-DAG $G$ and a coloring $c: \Sigma \to C$ with $|C| = k$ by incrementally constructing
    the trie $T = (U, \op{goto}, root)$ storing $L(c(G))$ and 
    the correspondence $\varphi \subseteq U\by V$ between vertices of $T$ and $G$. 
    (See the proof of \cref{thm:k:r:delta:const:fpt:maxmin})
}\label{algo:reduced:dag}
  \Procedure \proc{BuildColoredTrie}$(G = (V, E, s, t), c:\Sigma\to C)$\;
  \KwWork{A trie $T = (U = \bigcup_{d} U_d, \op{goto}, root)$\;}
  $root \gets$ a new vertex; 
  $U_0 \gets \set{root}$;
  $\varphi(root) \gets \set{s}$\; 
  $\op{goto} \gets \emptyset$; 
  $\op{visited} \gets \emptyset$\; 
  \For{$d := 1, \dots, r$}{
    $U_d \gets \emptyset$\; 
    \For (\brcomment{A vertex $x$ in $T$}) {$x \in U_{d-1}$ and $v \in \varphi(x)$}{
        \iIf{$\op{visited}(v) \not= \bot$}{
        \textbf{continue}\; 
        }
        $\op{visited}(v) \gets 1$ 
        \For  (\brcomment{An edge $e$ in $G$}) {$e = (v, a, w) \in \Eout(v)$}{
            $c \gets c(a)$\; 
            \If{$\op{goto}(x, c) = \bot$}{
              $y \gets$ a new vertex in $T$; 
              $\op{goto}(x, c) \gets y$; 
              $U_d \gets U_d \cup\set{ y }$\; 
            }
            $y \gets \op{goto}(x, c)$ \Comment*{a vertex $y$ in $T$}
            $\varphi(y) \gets \varphi(y) \cup\set{ w }$\; 
        }
    }
  }
  Let $T = (U, \op{goto}, root)$ be a trie with a vertex set $U = \bigcup_{d} U_d$ and an edge set $\op{goto}$, and the root\;  
  Return the $\Sigma$-DAG $H$ obtained from the trie $T$ by merging all leaves\; 
\end{algorithm}
\end{toappendix}

\begin{toappendix}
\label{sec:appendix:fpt}
\pagefigfptalgo=\thepage

\begin{figure}[t]
  \abovecaptionskip=-0.0\baselineskip
  \centering
  \vspace{-0.5\baselineskip}
  \includegraphics[width=0.9\textwidth,angle=0]{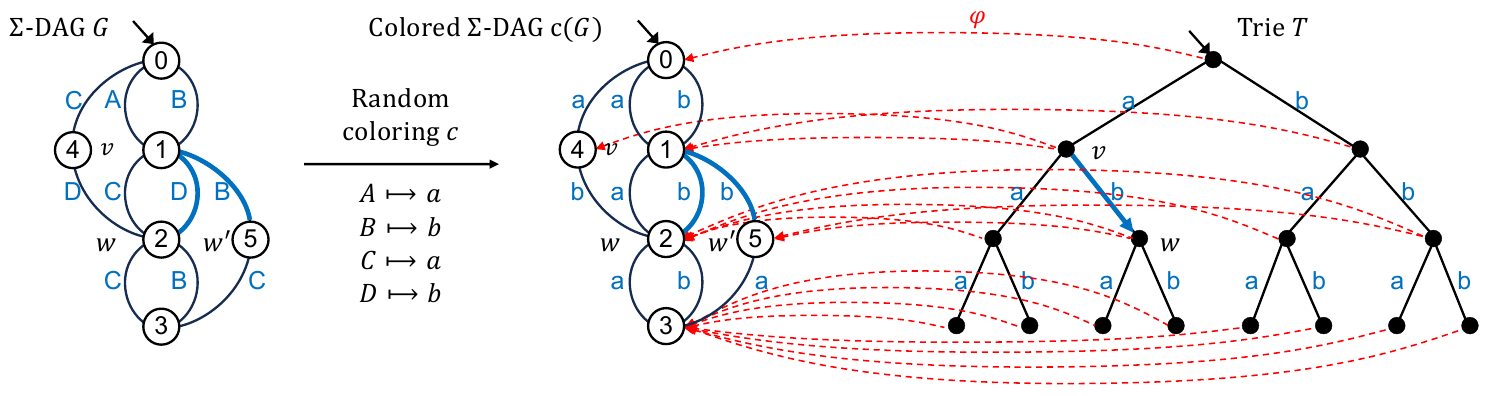}
  \caption{Illustration of
    the proof for \cref{lem:const:trie}, where dashed lines indicates a correspondence~$\varphi$.
  }\label{fig:fptalgo}
\end{figure}

\cref{fig:fptalgo} illustrates computation of reduced $C$-DAG $H$ from an input $\Sigma$-DAG $G$ over alphabet $\Sigma = \set{A, B, C, D}$ in \cref{lem:const:trie}, which shows $G$  (left), a random coloring $c$ on $C = \set{a,b}$, a colored $C$-DAG $c(G)$ (middle), and a reduced $C$-DAG $H$ in the form of trie $T$ (right). 
Combining \cref{lem:const:trie}, \cref{lem:k:const:dp:algo:maxmin}, and Alon \textit{et al.}~\cite{alon1995color}, we show the next theorem. 


\begin{theorem}\label{thm:k:r:delta:const:fpt:maxmin}
  When $r$ and $K$ are parameters,
  the \textsc{Max-Min Diverse String Set} on a $\Sigma$-DAG for $r$-strings is fixed-parameter tractable (FPT), where $\size(G)$ is
  an input. 
\end{theorem}

\begin{proof}
We show a sketch of the proof.
  We show a randomized algorithm using Alon \textit{et al.}'s color-coding technique~\cite{alon1995color}.
  Let $L(G)\subseteq \Sigma^r, k = rK$, and $C = [rK]$.
  We assume without loss of generality that $\Delta \le r$. 
  We randomly color edges of $G$ from $C$. 
  Then, we perform two phases below.
\begin{itemize}  
\item \textit{Preprocessing phase}: Using the FPT-algorithm of \cref{lem:const:trie}, reduce the colored $C$-DAG $c(G)$ with $\size(G)$ into another $C$-DAG $H$ with $L(H) = L(c(G)) \subseteq C^r$ and size bounded by $(rK)^r$. \cref{lem:const:trie} shows that this requires $t_\fn{pre} = O((rK)^r\cdot \size(G))$ time and space. 
  
\item \textit{Search phase}: Find a $\Delta$-diverse subset $\Y$ in $L(H)$ of size $|\Y| = K$ from $H$ using a modified version of \cref{algo:k:const:dp} in \cref{sec:algo} (details in footnote%
\footnote{This modification of \cref{algo:k:const:dp} is easily done at Line~7 of \cref{algo:k:const:dp} by replacing the term $\Ind{ lab(e_i)\not= lab(e_j) }$ with the term $\Ind{ \set{c(lab(e_i))\not= c(lab(e_j))} \land \set{lab(e_i)\not= lab(e_j)} }$. }). 
If such $\Y$ exists and $c$ is invertible, then $\X = c^{-1}(\Y)$ is a $\Delta$-diverse solution for the original problem. The search phase takes
\begin{math}
  t_\fn{search}
  = O(K^2 \Delta^{K^2} (rK)^{rK})
  =: g(K, r)
\end{math}
time, where $\Delta \le r$ is used. 
\end{itemize}

With the probability $p = (rK)!/(rK)^{rK} \ge 2^{-rK}$, for $C = [rK]$, the random $C$-coloring yields a colorful $\Delta$-diverse subset $\Y = c(\X) \subseteq L(H)$. Repeating the above process $2^{rK}$ times and derandomizing it using Alon \textit{et al.}~\cite{alon1995color} yields an FPT algorithm with total running time 
\begin{inalign}
  t = 2^{rK} r\log(rK) (t_\fn{pre} + t_\fn{search})
  = f(K, r, \Delta) \cdot \size(G),
\end{inalign}
where $f(K, r, \Delta) = O(2^{rK} r\log(rK)\cdot \{ (rK)^r + g(K, r) \})$ depends only on parameters $r$ and $K$. This completes the proof. 
\end{proof}

Similarly, we obtain the following result for Max-Sum Diversity.

\begin{theorem}\label{thm:k:r:delta:const:fpt:maxsum}
  When $r$ and $K$ are parameters, the \textit{Max-Sum Diverse String Set} on $\Sigma$-graphs for $r$-strings is fixed-parameter tractable (FPT), where $\size(G)$ is part of an input. 
\end{theorem}

\begin{proof}
  The proof proceeds by a similar discussion to the one in the proof of \cref{thm:k:r:delta:const:fpt:maxmin}.
  The only difference is the time complexity of $t_\fn{search}$. In the case of Max-Sum diversity, the search time of the modified algorithm in \cref{lem:k:const:dp:algo:maxsum} is $t_\fn{search} = O(\Delta K^2 M^K)$, where $M = \size(G)$. 
By substituting $M \le (rK)^k$ for $t_\fn{search}$, we have 
  \begin{math}
    t_\fn{search}
    = g'(K, r) \Delta, 
  \end{math}
  where $g'(K, r) := O(K^2 (rK)^{rK})$. 
  Since $g'(K, r)$ depends only on parameters,
  the claim follows. 
\end{proof}

\section{Hardness results}
\label{sec:hardness}

\begin{toappendix}
\label{sec:appendix:hardness}
\end{toappendix}

To complement the positive results in \cref{sec:algo} and \cref{sec:approx}, we show some negative results in classic and parameterized complexity.
In what follows, $\sigma = |\Sigma|$ is an alphabet size, $K$ is the number of strings to select, $r$ is the length of equi-length strings, and $\Delta$ is a diversity threshold. In all results below, we assume that $\sigma$ are constants,  and 
without loss of generality 
from \cref{rem:rep:strdag} 
that an input set $L$ of $r$-strings is explicitly given as the set itself.

\subsection{Hardness of Diverse String Set for Unbounded $K$}
\label{subsec:hard:div:strset}


Firstly, we observe the NP-hardness of \textsc{Max-Min} and \textsc{Max-Sum Diverse String Set} holds for unbounded $K$ even for constants $r\ge 3$.

\begin{theorem}[NP-hardness for unbounded $K$]\label{thm:nphard:div:strdag}
  When $K$ is part of an input, 
  \textsc{Max-Min} and \textsc{Max-Sum Diverse String Set} on $\Sigma$-graphs for $r$-strings are NP-hard even for any constant $r \ge 3$.
\end{theorem}

\begin{proof}
  We reduce an NP-hard problem 3DM~\cite{garey:johnson1983:intractability} to \textsc{Max-Min Diverse String Set} by a trivial reduction. Recall that given an instance consists of sets $A = B = C = [n]$ for some $n\ge 1$ and a set family $F \subseteq [n]^3$, and 3DM asks if there exists some subset $M \subseteq F$ that is a \textit{matching}, that is, any two vectors $X, Y \in M$ have no position $i \in [3]$ at which the corresponding symbols agree, i.e., $X[i] = Y[i]$. Then, we construct an instance of \textsc{Max-Min Diverse String Set} with $r = 3$ with 
  an alphabet $\Sigma = A\cup B\cup C$,
  a string set $L = F \subseteq \Sigma^3$,
  integers $K = n$ and $\Delta = r = 3$.
  Obviously, this transformation is polynomial time computable. 
  Then, it is not hard to see that for any $M\subseteq F$, $M$ is a matching if and only if $\Div[min]{d_H}(M) \ge \Delta$ holds.
  On the other hand, for \textsc{Max-Min Diverse String Set}, if we let $\Delta' = \binom{K}{2}$ then for any $M\subseteq F$, $M$ is a matching if and only if $\Div[sum]{d_H}(M) \ge \Delta'$ holds. Combining the above arguments, the theorem is proved. 
\end{proof}

\pagefighardmatch=\thepage
\begin{figure}[t]
\abovecaptionskip=0.5\baselineskip
\belowcaptionskip=0\baselineskip
\centering
\vspace{-1.25\baselineskip}
\nofbox{\includegraphics[height=4.9\baselineskip]{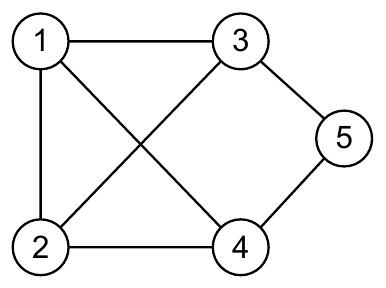}}
\hfil
\nofbox{\includegraphics[height=4.9\baselineskip]{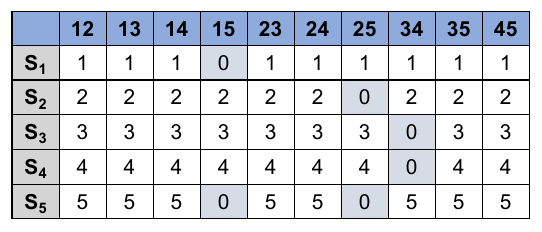}}
\caption{An example of reduction for the proof of \cref{thm:w1hard:div:strdag} in the case of $n = 5$, consisting of an instance $G$ of \textsc{Clique}, with a vertex set $V = \set{1,\dots,5}$ and a edge set $E \subseteq \sig E = \set{12, 13, \dots, 45}$ (left), and the associated instance $F = \set{S_1, \dots, S_n}$ of \textsc{Diverse $r$-String Set}, where $F$ contains $n=5$ $r$-strings with $r = |\sig E| = 10$ (right). Shadowed cells indicate the occurrences of symbol $0$.  
}\label{fig:hardmatch}
\end{figure}

We remark that 3DM is shown to be in FPT by Fellows, Knauer, Nishimura, Ragde, Rosamond, Stege, Thilikos, and Whitesides~\cite{fellows2008faster}.
Besides, we showed in \cref{sec:fpt} that \textsc{Diverse $r$-String Set} is FPT when parameterized with $K + r$ (\textsc{Max-Sum}) or $K + r + \Delta$ (\textsc{Max-Min}), respectively.
We show that the latter problem is W[1]-hard parameterized with~$K$. 

\begin{theorem}[\Wone-hardness of the string set and $\Sigma$-DAG versions for unbounded $K$]\label{thm:w1hard:div:strdag}
  When parameterized with $K$, \textsc{Max-Min} and \textsc{Max-Sum Diverse String Set} for a set $L$ of $r$-strings are W[1]-hard whether a string set $L$ is represented by either a string set $L$ or a $\Sigma$-DAG for $L$, where $r$ and $\Delta$ are part of an input. 
\end{theorem}

\begin{proof}
  We establish the W[1]-hardness of \textsc{Max-Min Diverse String Set} with parameter $K$ by reduction from \textsc{Clique} with parameter $K$. This builds on the NP-hardness of \textsc{$r$-Set packing} in
  Ausiello \textit{et al.}~\cite{ausiello1980structure}  with minor modifications (see also \cite{fellows2008faster}).
  Given a graph $G = (V, E)$ with $n$ vertices and a parameter $K \in \nat$, where $V = [n]$ and $E \subseteq \sig{E}$,
  we let $\sig{E} := \sete{ \set{i,j}  \mid i, j \in V, i\not= j }$. 
  We define the transformation $\phi_1$
  from $\pair{G, K}$ to $\pair{\Sigma, r, F, \Delta}$ and $\kappa(K) = K$ as follows.  
  We let  $\Sigma = [n]\withpos$,
  $r = |\sig{E}| = \binom{n}{2}$, and
  $\Delta = r$.
We view each $r$-string $S$ as a mapping $S: \sig{E} \to \Sigma$ assigning symbol $S(e) \in \Sigma$ to each unordered pair $e \in \sig{E}$. 
  We construct a family
  $F = \sete{ S_i \mid i \in V }$
  of $r$-strings such that 
  $G$ has a clique of $K$ elements
  if and only if 
  there exists a subset $M\subseteq F$ with 
    (a) size $|M| \ge \kappa(K) = K$, and
    (b) diversity $d_H(S, S') \ge r = \Delta$ for all distinct $S, S' \in M$ (*).  
Each $r$-string $S_i$ is defined based on the existence of the edges in $E$: 
(i) $S_i(e) = 0$ if $(i \in e)\land (e \not\in E)$, and
(ii) $S_i(e) = i$ otherwise.
By definition, $d_H(S_i, S_j) \le r$. We show that for any $i,j \in \sig E$, $S_i$ and $S_j$ have conflicts at all positions, i.e. $d_H(S_i,S_j) = r$, if and only if $\set{i,j} \in E$.
See \cref{fig:hardmatch} for example  of $F$. 
To see the correctness, suppose that $e = \set{i,j}\not\in E$. Then, it follows from $(i)$ that $S_i(e) = S_j(e) = 0$ since $(i\in e) \land (j \in e)\land e\not\in E$.
Conversely, if $e' = \set{i,j}\in E$, the condition (i) $S_i(e) = S_j(e) = 0$ does not hold for any $e \in \sig E$ because if $e \not= e'$, one of $(i \in e)$ and $(j \in e)$ does not hold, and if $e = e'$, $e\not\in E$ does not hold. This proves the claim (*). Since $\phi_1$ and $\kappa$ form an FPT-reduction. The theorem is proved.
\end{proof}

In this subsection, we show the hardness results of Diverse LCSs for unbounded $K$ in classic and parameterized settings by reducing them to those of \textsc{Diverse String Set} in~\cref{subsec:hard:div:strset}. 


\begin{theorem}\label{lem:fptreduce:strdag:to:lcs}
  Under Hamming distance, 
  \textsc{Max-Min} (resp.~\textsc{Max-Sum}) \textsc{Diverse String Set} for $m\ge 2$ strings parameterized with $K$ 
  is FPT-reducible to
  \textsc{Max-Min} (resp.~\textsc{Max-Sum}) \textsc{Diverse LCSs} for two string ($m = 2$) parameterized with $K$, where $m$ is part of an input. Moreover, the reduction is also a polynomial time reduction from \textsc{Max-Min} (resp.~\textsc{Max-Sum}) \textsc{Diverse String Set} to  \textsc{Max-Min} (resp.~\textsc{Max-Sum}) \textsc{Diverse LCSs}. 
\end{theorem}


We defer the proof of \cref{lem:fptreduce:strdag:to:lcs} in \cref{subsec:proof:thm:six:three}. 
Combining \cref{thm:nphard:div:strdag}, \cref{thm:w1hard:div:strdag}, and \cref{lem:fptreduce:strdag:to:lcs}, we have the corollaries. 

\begin{corollary}
  [NP-hardness]
  \label{thm:nphard:divlcs}
  When $K$ is part of an input, 
   \textsc{Max-Min} and \textsc{Max-Sum Diverse LCSs} for two $r$-strings are NP-hard, where $r$ and $\Delta$ are part of an input. 
\end{corollary}



\begin{corollary}[\Wone-hardness]
  \label{thm:w1hard:divlcs}
  When parameterized with $K$, 
  \textsc{Max-Min} and \textsc{Max-Sum Diverse LCSs} for two $r$-strings
  are W[1]-hard, where $r$ and $\Delta$ are part of an input. 
\end{corollary}


\subsubsection{Proof for \cref{lem:fptreduce:strdag:to:lcs}}
\label{subsec:proof:thm:six:three}

\begin{toappendix}


  In \cref{subsec:proof:thm:six:three}, we gave two key lemmas \cref{lem:w1hard:claim:matching} and \cref{lem:w1hard:claim:two}. These lemmas ensures that the LCSs of the constructed input strings $S_1$ and $S_2$ correctly encodes the original string set $L = \set{X_i}\idx i1s$ by the set of LCSs, namely, $LCS(S_1, S_2) = \sete{T_j = P_j\cdot X_j\cdot Q_j}\idx j1s$. 
  In this section, we show the diferred proofs for Claim~1 and Claim~2 in the proof for \cref{lem:w1hard:claim:two} in the proof for \cref{lem:w1hard:claim:two}. 

\end{toappendix}

In this subsection, we show the proof of \cref{lem:fptreduce:strdag:to:lcs}, which is deferred in the previous section. 
Suppose that we are given an instance of \textsc{Max-Sum Diverse String Set} consisting of  integers $K, r\ge 1$, $\Delta \ge 0$, and any 
set 
$L = \set{X_i}\idx i1s\subseteq \Sigma^r$ 
of $r$-strings, where $s = |L| \ge 2$.
We let 
$\Xi = \sete{ a_{i,j}, b_{i,j} \mid i,j \in [s] }$ be a set of mutually distinct symbols, and 
$\Gamma = \Sigma\cup\Xi$ be a new alphabet with $\Sigma\cap\Xi = \emptyset$.
We let 
$\sig T = \set{T_i:= P_i X_i Q_i}\idx i1s$ 
be the set of $s$ strings of length $|T_i| = r + 2s$ over $\Gamma$, where 
\begin{math}
P_i := a_{i1}\dots a_{is} \in \Gamma^s, 
Q_i := b_{i1}\dots b_{is} \in \Gamma^s, 
\forall i \in [s] 
\end{math}. 


Now, we construct two input strings $S_1$ and $S_2$ over $\Gamma$ in an instance of \textsc{Max-Min Diverse LCSs} so that $LCS(\S_1, \S_2) = \sig T$. 
For all $i \in [s]$, we factorize each  strings $T_i$ of length $(r+2s)$ into three substrings $A_i, W_i, B_i \in \Gamma^+$, called \textit{segments}, such that $T_i = A_i\cdot W_i\cdot B_i$ such that 
(i) We partition $P_i$ into $P_i = A_i\cdot\bar{A_i}$, where $A_i := P_i[1..s-i+1]$ is the prefix with length $s-i+1$ and $\bar{A_i} = P_i[s-i+2..s]$ is the suffix with length $i-1$ of $P_i$. 
(ii) We partition $Q_i$ into $Q_i = \bar{B_i}\cdot B_i$, where 
$\bar{B_i} = Q_i[1..s-i]$ is the prefix with length $s-i$ and 
$B_i:= Q_i[s-i+1..s]$ is the suffix with  length $i$ 
of $Q_i$.
(iii) We obtain a string $W_i := \bar{A_i}\cdot X_i \cdot \bar{B_i}$ with 
length $r + s - 1$ from $X_i$ by prepending and appending $\bar{A_i}$ and $\bar{B_i}$ to $X_i$. 
Let 
$\sig A = \sete{ A_i }\idx i1s$, 
$\sig B = \sete{ B_i }\idx i1s$, and 
$\sig W = \sete{ W_i }\idx i1s$
be the \textit{groups} of the segments of the \textit{same types}. 
See \cref{fig:redlcs2} for examples of $\sig A, \sig B$, and $\sig W$. 
Then, we define the set $\S = \set{S_1, S_2}$ of two input strings $S_1$ and $S_2$ of the same length
\begin{math}
|S_1| = |S_2| 
= s(r+2s)
\end{math} 
by: 
\begin{align}
  \S_1
  &= \textstyle \prod_{i=1}^{s} A_i\cdot \prod_{i=1}^{s} W_i\cdot \prod_{i=1}^{s} B_i
  = (A_1\cdots A_s) \cdot (W_1\cdots W_s) \cdot (B_1\cdots B_s), 
  \nonumber 
  \\\S_2 
  &= \textstyle \prod_{i=s}^{1} T_i 
  = \prod_{i=s}^{1} (A_i\cdot W_i\cdot B_i) 
  = (A_s\cdot W_s\cdot B_s)\cdots (A_1\cdot W_1\cdot B_1). 
  \label{eq:two:inputstrings}
\end{align}

\cref{fig:redlcs} shows an example of  $\S$ for $s = 4$. 
We 
observe the following properties of  $\S$: 
(P1) $S_1$ and $S_2$ are segment-wise permutations of each other; 
(P2) if all segments in any group $\sig Z = \set{Z_i}\idx i1s \in \set{\sig A, \sig B, \sig W}$ occur one of two input strings, say $S_1$, in the order $Z_1, \dots, Z_s$, then they occur in the other, say $S_2$, in the reverse order $Z_s, \dots, Z_1$;  
(P3) $A_i$'s (resp.~$B_i$'s) appear in $S_2$ from left to right in the order $A_s, \dots, A_1$ (resp.~$B_s, \dots,  B_1$); 
(P4) $\sig A$ and $\sig B$ satisfy 
$|A_1| > \dots > |A_s|$ and 
$|B_1| < \dots < |B_s|$. 
We associate a bipartite graph $\sig B(\S) = (V = V_1\cup V_2, E)$ to $\S$, where 
(i) $V_k$ consists of all positions in $S_k$ for $k=1,2$, and 
(ii) $E\subseteq V_1\by V_2$ is an edge set such that 
$e = (i_1, i_2) \in E$ if and only if  both ends of $e$ have the same label $S_1[i_1] = S_2[i_2] \in \Sigma$. 
Any sequence $M = ((i_k, j_k))\idx k1\ell \in E^\ell$ of $\ell$ edges 
is an \textit{(ordered) matching} if 
$i_1 \not= j_1$ and $i_2 \not= j_2$, and is \textit{non-crossing} if 
$(i_1 < j_1)$ and $(i_2 < j_2)$. 

\begin{figure}
  \centering
  \centering
    \subcaptionbox{The set $\sig T$\label{fig:redlcs2}}{\includegraphics[height=3.1cm]{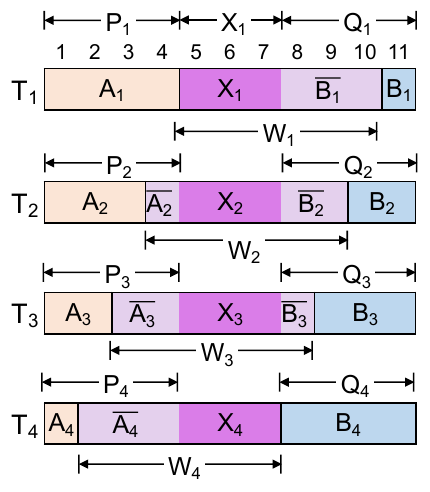}}
    \subcaptionbox{Input strings $S_1$ and $S_2$\label{fig:redlcs}}{\includegraphics[height=3.2cm]{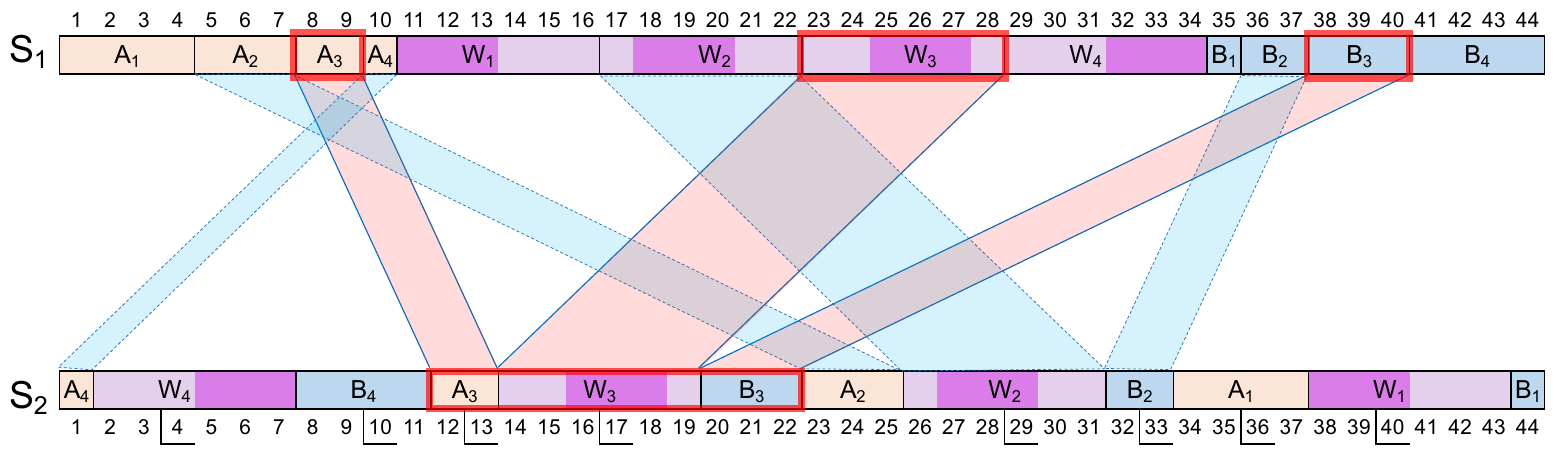}}
    \caption{
      Construction of the FPT-reduction from \textsc{Max-Min Diverse String Set} to \textsc{Max-Min Diverse LCS} in the proof of \cref{lem:fptreduce:strdag:to:lcs}, where $s = 4$.
      We show (a) the set $\sig T$ of $s$ $r$-strings and (b) a pair of input strings $S_1$ and $S_2$. 
      Red and blue parallelograms, respectively, indicate allowed and prohibited matchings between the copies of blocks $T_3 = A_3 W_3 B_3$ in $S_1$ and $S_2$.
    }
\end{figure}

\begin{lemma}\label{lem:w1hard:claim:matching}
For any $M \subseteq V_1\by V_2$ and any $\ell \ge 0$, $\sig B(\S)$ has a non-crossing matching $M$ of size $\ell$ if and only if $S_1$ and $S_2$ have a common subsequence $C$ with length $\ell$ of $S_1$ and $S_2$. Moreover, the length $\ell=|M|$ is maximum if and only if $C \in LCS(S_1, S_2)$. 
\end{lemma}

\begin{proof}
If there exists a non-crossing matching $M = \sete{(i_k,j_k) \mid k \in [\ell]} \subseteq E$ of size $\ell\ge 0$, 
we can order the edges in the increasing order such that $i_{\pi(1)} < \dots < i_{\pi(\ell)}$, $j_{\pi(1)} < \dots < j_{\pi(\ell)}$ for some permutation $\pi$ on $[\ell]$. Then, the string $S_1(M) := S_1[i_{\pi(1)}]\cdots S_1[i_{\pi(\ell)}] \in \Sigma^\ell$ (equivalently, $S_2(M) := S_2[j_1]\cdots S_2[j_\ell]$) forms the common subsequence associated to $M$.     
\end{proof}

In \cref{lem:w1hard:claim:matching}, we call a non-crossing ordered matching $M$ associated with a common subsequence $C$ a \textit{matching labeled with} $C$. 
We show the next lemma.

\begin{lemma}\label{lem:w1hard:claim:two}
  $LCS(S_1, S_2) = \sete{ T_j \mid i\in [s]}$, where $T_j = P_j\cdot X_j\cdot Q_j$ for all $j \in [s]$. 
\end{lemma}

\begin{proof}
We first observe that each segment $Z \in \Sigma^+$ in each group $\sig Z$ within $\set{\sig A, \sig B, \sig W}$ occurs exactly once in each of $S_1$ and $S_2$, respectively, as a consecutive substring. 
Consequently, For each $Z$ in $\sig Z$, $\sig B(\S)$ has exactly one non-crossing matching 
$M_Z$
\textit{labeled with $Z$ connecting} the occurrences of $Z$ 
in $S_1$ and $S_2$.
From (P2), we show the next claim. 

\textit{Claim}~1: 
If $\sig B(\S)$ contains any inclusion-wise maximal non-crossing matching $M_*$, it connects exactly one segment $Z$ from each of three groups $\sig A, \sig B$, and $\sig W$. 

\begin{toappendix}
\par\medskip
\textit{Claim}~1: 
If $\sig B(\S)$ contains any inclusion-wise maximal non-crossing matching $M_*$, it connects exactly one segment $Z$ from each of three groups $\sig A, \sig B$, and $\sig W$. 
\par 
(Proof for Claim~1)
$M$ can be labeled by at most one segment from each segment group $\sig Z \in \set{\sig A, \sig B, \sig W}$. 
To see this, suppose to contradicts that $M$ contains submatches $M_i$ and $M_j$ labeled with two distinct segments, $Z_i$ and $Z_j$ with $i < j$, respectively, from the same group $\sig Z$. 
By construction, all segments in $\sig Z$ occur one of two input strings, say $S_1$, in the increasing order $Z_1, \dots, Z_s$, then they occur in the other, say $S_2$, in the decreasing order $Z_s, \dots, Z_1$ (see \cref{fig:redlcs}). 
Then, it follows that the submatchings $M_i$ and $M_j$ must cross each other. From \cref{lem:w1hard:claim:matching}, this  contradicts 
that $M$ is non-crossing. 
(End of Proof for Claim~1)
\end{toappendix}

From Claim~1, we assume that a maximum (thus, inclusion-maximal) non-crossing matching $M_*$ contains submatches labeled with segments $A_i, W_j, B_k$ one from each group in any order, where $i,j,k\in [s]$. Then, $M$ must contain $A_i, W_j, B_k$ in this order, namely, $A_i\cdot W_j\cdot B_k \in CS(S_1, S_2)$ because some edges cross otherwise (see \cref{fig:redlcs}). Therefore, we have that the concatenation $T_{j_*} := A_{j_*}\cdot W_{j_*}\cdot B_{j_*}$ belongs to $CS(S_1, S_2)$, and it always has a matching in $\sig B(S_1,S_2)$. 
From (P3) and (P4), we can show the next claim. 

\textit{Claim~2}: If $M_*$ is maximal and contains $A_i\cdot W_{j_*}\cdot B_k$, then 
$i=j_*=k$ holds.

\begin{toappendix}
\par\medskip
\textit{Claim~2}: If $M_*$ is maximal and contains $A_i\cdot W_{j_*}\cdot B_k$, then 
$i=j_*=k$ holds.
\par 
(Proof for Claim~2) 
Suppose to contradict that $\sig B(\S)$ contains a maximal match $A_i\cdot W_{j_*}\cdot B_k$ with $i,j,k$. Without loss of generality, we fix the submatch for $W_{j_*}$. 
If $i\not= j_*$ holds, the submatch $M_{A_i}$ for $A_i$ (resp.~$M_{B_k}$ for $B_k$) must appear in $\sig B$ to the left (resp.~right) of the submatch $M_{W_{j_*}}$ for $W_{j_*}$. 
Since $A_i$'s (resp.~$B_i$'s) appear in $S_2$ from left to right in the order $A_s, \dots, A_1$ (resp.~$B_s, \dots,  B_1$), this implies that $i > j_*$ (resp.~$j_* > k$). 
By construction, $i > j_*$ implies $|A_i| < |A_{j_*}|$ (resp.~$j_* > k$ implies $|B_{j_*}| > |B_k|$). 
Therefore, $|A_i\cdot W_{j_*}\cdot B_k| < |A_{j_*}\cdot W_{j_*}\cdot B_{j_*}|$ unless $i =j_* = k$. (End of Proof for Claim~2)
\end{toappendix}

From Claim~2, we conclude that $T_{j_*} = A_{j}\cdot W_{j}\cdot B_{j}$ is the all and only members of $LCS(S_1,S_2)$ for all $j \in [s]$. 
Since $A_{j}\cdot W_{j}\cdot B_{j} = P_{j}\cdot X_{j} \cdot B_{j} = T_{j}$, the lemma is proved. 
\end{proof}

Using~\cref{lem:w1hard:claim:two}, we finish the proof for \cref{lem:fptreduce:strdag:to:lcs}. 

\begin{myproofonly}{Proof for \cref{lem:fptreduce:strdag:to:lcs}}
Recall that integers $K, r\ge 1$, $\Delta \ge 0$, and a string set $L = \set{X_1, \dots, X_s}\subseteq \Sigma^r$ of $r$-strings form an instance of \textsc{Max-Min Diverse String Set}. 
Let $\Delta' := \Delta + 2s$, $K' = \kappa(K) := K$, and $\S = \set{S_1, S_2} \subseteq \Gamma^*$ be the associated instance of \textsc{Max-Min Diverse LCS} for two input strings. 
Since the parameter $\kappa(K) = K$ depends only on $K$, it is obvious that this transformation can be computed in FPT. We show that this forms an FPT-reduction from the former problem to the latter problem. 
By \cref{lem:w1hard:claim:two},  we have the next claim. 

\textit{Claim}~3: For any $i,j \in [s]$, $d_H(T_i, T_j) = d_H(X_i, X_j) + 2s$. 

(Proof for Claim~3)
By \cref{lem:w1hard:claim:two},  $LCS(S_1, S_2) = \sete{ T_j \mid i \in [K]}$.
By construction, 
$T_j = P_j \cdot X_j \cdot Q_j$ and 
$|P_j| = |Q_j| = s$, and $d_H(P_i, P_j) = d_H(Q_i, Q_j) = s$ for any $i,j\in [s],\; i\not= j$. Therefore, we can decompose $d_H(T_i, T_j)$ by 
\begin{math}
  d_H(T_i, T_j)
= d_H(P_i, P_j) + d_H(X_i, X_j) +  d_H(Q_i, Q_j)
= d_H(X_i, X_j) + 2s
\end{math}
(End of Proof for Claim~3)

Suppose that $\Y\subseteq LCS(S_1, S_2)$ is any feasible solution such that $|\Y| = K'$ for \textsc{Max-Sum Diverse LCSs}, where $K' = K$. 
From~\cref{lem:w1hard:claim:two}, we can put $\Y = \set{T_{i_j}}_{j \in J}$ for some  $J \subseteq [s]$.
From Claim~3, we can see that 
\begin{math}
    \Div[min]{d_H}(\Y)
    = \Div[min]{d_H}(\X) + 2s, 
\end{math}
where $\X = \set{X_j}_{j \in J}$ is a solution for \textsc{Max-Min Diverse String Set}. 
Thus, $\Div[min]{d_H}(\X) \ge \Delta$ if and only if $\Div[min]{d_H}(\Y) \ge \Delta + 2s = \Delta'$. This shows that the transformation properly forms NP- and FPT-reductions. 
To obtain NP- and FPT-reductions for the \textsc{Max-Sum} version, we keep $K$ and $\S = \set{S_1, S_2}$ in the previous proof, and modify $\Delta' := \Delta + 2s\binom{K}{2}'$, 
where $\binom{K}{2}'$ $\!:= $ $\set{ (i,j)\in \binom{K}{2}\mid i \not= j }$. 
From Claim~3, we have that $\Div[sum]{d_H}(\Y) = \Div[sum]{d_H}(\X) + 2s\binom{K}{2}'\!$, 
and the correctness of the reduction immediately follows. 
Combining the above arguments, the theorem is proved. 
\end{myproofonly}




\section*{Acknowledgements}

The authors express sincere thanks  to anonymous reviewers 
for their valuable comments, 
which significantly improved the presentation and quality of this paper.
The last author would like to thank Norihito Yasuda, Tesshu Hanaka, Kazuhiro Kurita, Hirotaka Ono of AFSA project, and Shinji Ito for fruitful discussions and helpful comments.




\clearpage
%
\bibliography{ref}

\end{document}